\pdfoutput=1
\documentclass[runningheads,orivec]{llncs}
\usepackage{graphicx} 
\usepackage{lineno}
\usepackage[T1]{fontenc}
\usepackage{amssymb}
\usepackage{tikz-cd}
\usepackage{array}
\usepackage{mdframed}
\usepackage{thm-restate}
\usepackage{pgf-pie}
 \usepackage{enumitem}

\usepackage{underscore}
\usepackage{pgfornament}
\usepackage{scalerel}
\usepackage{MnSymbol}
\usepackage{csquotes}
\usepackage{colortbl}
\usepackage{hyperref}
\usepackage{cleveref}
\usepackage{color}

\usepackage{afterpage} 
\usepackage{hhline}
\usepackage{multirow}
\usepackage{longtable}
\usepackage{adjustbox}
\usepackage[tableposition=below]{caption}

\usepackage[T1]{fontenc}
\usepackage[ttdefault=true]{AnonymousPro}

\newcommand{\A}{\ensuremath{\mathcal{A}}}

\newcommand{\B}{\ensuremath{\mathcal{B}}}

\newcommand{\C}{\ensuremath{\mathcal{C}}}

\newcommand{\T}{\ensuremath{\mathcal{T}}}

\newcommand{\s}{\ensuremath{\sigma}}

\renewcommand{\P}{\ensuremath{\mathcal{P}}}

\newcommand{\N}{\ensuremath{\mathbb{N}_{\omega}}}

\newcommand{\qf}[1]{\ensuremath{QF(#1)}}

\newcommand{\eq}[1]{\ensuremath{Eq(#1)}}

\newcommand{\ax}[1]{\ensuremath{Ax(#1)}}

\newcommand{\minmod}{\ensuremath{\textbf{minmod}}}

\newcommand{\Forall}[1]{\ensuremath{\forall\,#1.\:}}

\newcommand{\Exists}[1]{\ensuremath{\exists\,#1.\:}}

\newcommand{\Teven}{\ensuremath{\T_{even}^{\infty}}}

\newcommand{\Tninfty}{\ensuremath{\T_{n,\infty}}}

\newcommand{\bb}{\ensuremath{\varsigma}}

\newcommand{\Tbb}{\ensuremath{\T_{\bb}}}

\newcommand{\Tone}{\ensuremath{\T_{\mathrm{I}}}}

\newcommand{\overarrow}[1]{\ensuremath{\overrightarrow{#1}}}

\newcommand{\vars}{\ensuremath{vars}}

\newcommand{\Tgeqn}{\ensuremath{\T_{\geq n}}}

\newcommand{\Tinfty}{\ensuremath{\T_{\infty}}}

\newcommand{\Tmn}{\ensuremath{\T_{\langle m,n\rangle}}}

\newcommand{\Tbbs}{\ensuremath{\T^{s}_{\bb}}}

\newcommand{\wit}{\ensuremath{wit}}

\newcolumntype{P}[1]{>{\centering\arraybackslash}p{#1}}

\newcommand{\F}{\ensuremath{\mathcal{F}}}

\renewcommand{\P}{\ensuremath{\mathcal{P}}}

\newcommand{\NNEQ}[1]{\ensuremath{\neq(#1_{1},\ldots,#1_{n})}}

\newcommand{\NNNEQ}[2]{\ensuremath{\neq(#1_{1},\ldots,#1_{#2})}}

\newcommand{\nocontentsline}[3]{}
\newcommand{\toclesslab}[3]{\bgroup\let\addcontentsline=\nocontentsline#1{#2\label{#3}}\egroup}

\newcommand{\LIM}{\ensuremath{\lim_{n\rightarrow\infty}}}

\newcommand{\An}{\ensuremath{A_{n}}}

\renewcommand{\AA}[1]{\ensuremath{A_{#1}}}

\newcommand{\fr}[2]{\ensuremath{#1/#2}}

\newcommand{\No}{\ensuremath{\mathbb{N}^{*}}}

\newcommand{\spec}[1]{\ensuremath{\textit{Spec}(#1)}}

\newcommand{\specn}[1]{\ensuremath{\textit{Spec}_{n}(#1)}}

\newcommand{\Spec}{\ensuremath{\textit{Spec}}}

\newcommand{\bv}{\ensuremath{\textbf{BV}[n]}}

\newcommand{\dom}[1]{\ensuremath{\textit{dom}(#1)}}

\newcommand{\Tpt}{\ensuremath{\T_{=2^{i}}}}

\newcommand{\Tnpt}{\ensuremath{\T_{\neq 2^{i}}}}

\newcommand{\mf}{\ensuremath{G}}

\let\tp\texorpdfstring

\begin{document}

\title{Number theory combination: \\ natural density and SMT}

\author{Guilherme Toledo and Yoni Zohar}
\institute{Bar-Ilan University, Israel}

\maketitle              
\begin{abstract}
The study of theory combination in Satisfiability Modulo Theories (SMT) involves various model theoretic properties (e.g., stable infiniteness, smoothness, etc.).
We show that 
such properties 
can be partly captured by the natural density of the spectrum of the studied theories, which is the set of sizes of their finite models. 
This enriches the toolbox of the theory combination researcher, by providing
new tools to determine the possibility of combining theories.
It also reveals interesting and surprising connections between theory combination 
and number theory.

\end{abstract}

\section{Introduction}

Imagine this:
you are a researcher in Satisfiability Modulo Theories (SMT)~\cite{BSST21}, 
studying a theory $\T$, which is the combination of 
theories $\T_{1}$ and $\T_{2}$ 
(in the same way that, say, the theory of lists of integers is the combination of the theories of lists and integers). 
Given algorithms for $\T_{1}$ and $\T_{2}$, 
you can plug them together using theory combination methods, such as
Nelson and Oppen's method~\cite{NelsonOppen}, 
polite combination \cite{ranise:inria-00000570}, 
shiny combination \cite{DBLP:journals/jar/TinelliZ05}, 
or gentle combination~\cite{gentle}. 

But, before you 
can produce a decision procedure for $\T$,
 you must test certain properties of $\T_{1}$ and $\T_{2}$, or their absence, 
to determine their applicability to the combination method.
For example, using the Nelson-Oppen method requires 
that both theories are stably infinite, while using
the polite combination method requires that one of 
them is strongly polite.
The obvious way of doing so is by directly applying the definitions of these properties, what can be highly non-trivial (for example, to prove a theory is strongly polite, one needs to construct a computable function satisfying an involved set of conditions).

In this paper, we give you alternative tests, based on
number theoretic
 {\em natural densities}~\cite{Tenenbaum},
computed over the {\em spectrum} of the theory~\cite{descriptivecomplexity}.
When testing whether a theory admits or lacks a theory combination property,
you can now use these tests.
We provide examples for cases where this is simpler
to do, compared to the direct application of the definitions.
Beyond the introduction of such tools,
the results of this paper relate number theory and theory combination
in surprising and insightful ways.
We focus on one-sorted theories, leaving many-sorted ones for future work.

\Cref{preliminaries} surveys relevant notions.  
\Cref{computability and density} contains our main results:
sufficient and necessary conditions
for theory combination properties, in terms of the natural density.
In \Cref{non empty} we provide  generalizations to non-empty signatures.
\Cref{conclusion} summarizes, and gives directions for future research.


\section{Preliminaries}\label{preliminaries}

If $X$ is a set, $|X|$ denotes its cardinality.
We denote by $\aleph_{0}$ the cardinality of $\mathbb{N}$, which for us contains $0$;
the set $\mathbb{N}\setminus\{0\}$ is denoted by $\No$.

\subsection{First-order logic}\label{many sorted}

A first-order {\em signature} $\Sigma$ is a pair $(\F_{\Sigma}, \P_{\Sigma})$, where: $\F_{\Sigma}$ is a countable set of function symbols, each with an arity $n\in\mathbb{N}$; and $\P_{\Sigma}$ is a countable set of predicate symbols, each with an arity $n\in\mathbb{N}$, containing at least the equality $=$ of arity $2$. 
We denote by $\Sigma_{1}$ the signature with no function or predicate symbols other than $=$, which is therefore called {\em empty}. Assuming countably many variables, we define by structural induction {\em terms}, {\em literals}, {\em formulas}, and {\em sentences} (formulas without free variables) in the usual way. 
The set of all quantifier-free $\Sigma$-formulas is denoted by $\qf{\Sigma}$; the set of all variables in $\varphi$ shall be written as $\vars(\varphi)$. 

A {\em $\Sigma$-interpretation} $\A$ consists of: a non-empty set $\dom{\A}$, called the {\em domain} of $\A$; for each function symbol $f$ of arity $n$, a function $f^{\A}:\dom{\A}^{n}\rightarrow\dom{\A}$; for each predicate symbol $P$ of arity $n$, a subset $P^{\A}$ of $\dom{\A}^{n}$, where $=^{\A}$ is the identity; and, for every variable $x$, an element $x^{\A}$ of $\dom{\A}$. The value of a term $\alpha$ in $\A$ is denoted by $\alpha^{\A}$, while for a set of terms $\Gamma$ we make $\Gamma^{\A}=\{\alpha^{\A} : \alpha\in \Gamma\}$; if $\A$ satisfies the formula $\varphi$, we write $\A\vDash \varphi$. Recurrent formulas include those in \Cref{card-formulas}, that are satisfied by an interpretation $\A$ iff: $\A$ has at least $n$ elements, in the case of $\psi_{\geq n}$; $\A$ has at most $n$ elements, in the case of $\psi_{\leq n}$; and $\A$ has precisely $n$ elements, in the case of $\psi_{=n}$. 

\begin{figure}[t]
\begin{mdframed}
\begin{equation*}
\begin{aligned}
    \NNEQ{x}=&\bigwedge_{i=1}^{n-1}\bigwedge_{j=i+1}^{n}\neg(x_{i}=x_{j})\\
    \psi_{\geq n}=&\Exists{{x_1\ldots x_n}}\NNEQ{x}
\end{aligned}
\quad\quad\quad
\begin{aligned}
\psi_{\leq n}=&\Exists{x_1,\ldots,x_n}\Forall{y}\bigvee_{i=1}^{n}y=x_{i}\\
\psi_{=n}=&\psi_{\geq n}\wedge\psi_{\leq n}
\end{aligned}
\end{equation*}
\end{mdframed}
\caption{Cardinality formulas.}
\label{card-formulas}
\end{figure}

A \textit{theory} is the class of all interpretations (thus called $\T$-interpretations, or the models of $\T$) satisfying some set of sentences $\ax{\T}$ (which does not need to be computably enumerable), called the {\em axiomatization} of $\T$.
A formula $\varphi$ is then: 
{($\T$-)}\textit{satisfiable} if there is a 
{($\T$-)}interpretation that satisfies $\varphi$; 
{($\T$-)}\textit{equivalent} to a formula $\psi$ if every 
{($\T$-)}interpretation that satisfies one also satisfies the other; and 
{($\T$-)}\textit{valid} if every ($\T$-)interpretation satisfies $\varphi$, denoted $\vDash\varphi$ 
($\vDash_{\T}\varphi$). 

We denote, for $n\leq m$, the set $\{n, \ldots, m\}$ by $[n,m]$; if $n=0$, we simplify it to $[m]$. 
Of course, $|[n,m]|=m-n+1$, and $|[m]|=m+1$.
Furthermore, $A\cap[1,n]$ will be denoted by $\An$; we denote $\{|\dom{\A}| : \text{$\A$ is a $\T$-interpretation}\}\cap\mathbb{N}$ by $\spec{\T}$, and we define $\specn{\T}$ as $\spec{\T}\cap [1,n]$. 
Analogously, 
$\spec{\T,\phi}$
is the set of finite cardinalities of $\T$-interpretations that satisfy $\phi$.
We can then also define $\specn{\T,\phi}$ as $\spec{\T,\phi}\cap[1,n]$. 

\subsection{Number theory}\label{number theory}

The {\em natural density}~\cite{Tenenbaum} of a set $A\subseteq \mathbb{N}$  is the following real number,
if it exists (and then we say the density of $A$ is well-defined):
$\mu(A)=\LIM |A\cap [n]|/|[n]|$.

\begin{example}
        Consider the set $A$ of even non-negative integers: we then have that $\mu(A)$ is the limit of the sequence $a_{n}$ which equals $\fr{(n+2)}{2(n+1)}$ if $n$ is even, and $\fr{1}{2}$ if it is odd, meaning that $\mu(A)$ is well-defined and equals $\fr{1}{2}$.
\end{example}

It is easy to prove that $\mu$ satisfies, for all disjoint sets $A$ and $B$ for which it is defined: $0\leq \mu(A)$; $\mu(\mathbb{N})=1$; and $\mu(A\cup B)=\mu(A)+\mu(B)$.
The subsets of the non-negative integers we shall calculate the natural density of are sets of finite cardinalities of interpretations in a theory: since they are never zero (as we assume $\dom{\A}$ is never empty), we can change $\mu(A)$ to be the limit of the ratio of $|A\cap\{1,\ldots,n\}|$ to $|\{1,\ldots,n\}|=n$.\footnote{Of course, this does not change the value of $\mu(A)$.}
With this, we can finally define the natural density of a theory (relative to a quantifier-free formula or not) as the natural density of its spectrum: $\mu(\T)=\LIM \fr{|\specn{\T}|}{|[1,n]|}$, and $\mu(\T,\phi)=\LIM \fr{|\specn{\T,\phi}|}{|[1,n]|}$.

\begin{definition}\label{definition computable number}
    Let $r\in\mathbb{R}$. $r$ is computable~\cite{Turing} if there are computable sequences 
    $\{a_{n}\}_{n\in\mathbb{N}}$ in $\mathbb{Z}$ and $\{b_{n}\}_{n\in\mathbb{N}}$ in $\No $ with $\LIM \fr{a_{n}}{b_{n}}=r$.
\end{definition}

\begin{example}\label{computable number examples}~
          Every rational number $\fr{p}{q}$ is computable: just take $a_{n}=p$ and $b_{n}=q$.
          The number $\sum_{n=1}^{\infty}2^{-\bb(n)}=0.57824...$ is not computable, for $\bb$ the busy beaver function~\cite{Rado}, which
          maps $n\in\mathbb{N}$ to the maximum number of $1$’s a Turing machine with at most $n$ states can write when it halts, assuming the tape begins with only $0$’s.
          Now, $0.57824...$.
          is the limit of $\fr{5}{10}, \fr{57}{100}, \fr{578}{1000},\ldots$.
          Consider then the theory $\T$ with models of size $1$ through $5$, $11$ through $52=57-5$, $101$ through $521=578-57$, and so on.
          Its density is the limit of the fractions $\fr{5}{10}$, $\fr{57}{100}$, $\fr{578}{1000}$ and so on, i.e. $0.57824...$, although this number is irrational.
          More generally, any $0\leq r\leq 1$ is the density of some theory.
\end{example}

\subsection{Theory combination}\label{theory combination}
In what follows, let $\Sigma$ be an arbitrary signature and $\T$ be a $\Sigma$-theory.

$\T$ is \textbf{stably infinite}~\cite{NelsonOppen} if for every satisfiable quantifier-free formula $\phi$, there is a $\T$-interpretation $\A$ that satisfies $\phi$ with $|\dom{\A}|\geq \aleph_{0}$. $\T$ is \textbf{smooth} when, for all quantifier-free formulas $\phi$, $\T$-interpretations $\A$ that satisfy $\phi$, and cardinals $\kappa>|\dom{\A}|$, there exists a $\T$-interpretation $\B$ that satisfies $\phi$ with $|\dom{\B}|=\kappa$.
Notice that being smooth implies being stably infinite.

\begin{example}
    The theory axiomatized by $\{\psi_{\geq 3}\}$ is smooth, as we can always add more elements to an interpretation.
\end{example}

$\T$ is \textbf{finitely witnessable}~\cite{RanRinZar} when there is a computable function $\wit$ (called a witness) from the quantifier-free formulas into themselves such that, for every quantifier-free formula $\phi$: $(\textit{I})$ $\phi$ and $\Exists{\overarrow{x}}\wit(\phi)$ are $\T$-equivalent, where $\overarrow{x}=\vars(\wit(\phi))\setminus\vars(\phi)$; and $(\textit{II})$ if $\wit(\phi)$ is $\T$-satisfiable, then there is a $\T$-interpretation $\A$ that satisfies $\wit(\phi)$ and, in addition, $\dom{\A}=\vars(\wit(\phi))^{\A}$ (that is, every element of $\dom{\A}$ is the interpretation of a variable in $\wit(\phi)$). 
Now, given a finite set of variables $V$ on the signature $\Sigma$, and an equivalence $E$ on $V$, the \textbf{arrangement} on $V$ induced by $E$, written $\delta_{V}^{E}$ or $\delta_{V}$ if $E$ is clear from context, is the formula $\bigwedge_{xEy}(x=y)\wedge\bigwedge_{x\overline{E}y}\neg(x=y)$, where $\overline{E}$ is the complement of $E$.
Intuitively, an arrangement codifies the relationships between a finite set of variables, that is, if they should be equal or different to one another.
$\T$ is then \textbf{strongly finitely witnessable}~\cite{JB10-LPAR} if it has a witness $\wit$ (that in this case will be called a strong witness) satisfying, in addition to $(\textit{I})$ and $(\textit{II})$, the stronger $(\textit{II}^{*})$: for every finite set of variables $V$ and arrangement $\delta_{V}$ on $V$, if $\wit(\phi)\wedge\delta_{V}$ is $\T$-satisfiable, then there exists a $\T$-interpretation $\A$ that satisfies that formula and, in addition, $\dom{\A}=\vars(\wit(\phi)\wedge\delta_{V})^{\A}$.

\begin{example}
    The theory axiomatized by $\{\psi_{\leq 3}\}$ has as strong witness $\wit(\phi)=\phi\wedge\bigwedge_{i=1}^{3}x_{i}=x_{i}$, where $x_{1}$, $x_{2}$ and $x_{3}$ are fresh variables (i.e., not in $\phi$).
\end{example}

%
$\T$ has the \textbf{finite model property (FMP)} if for every $\T$-satisfiable quantifier-free formula $\phi$, there is a $\T$-interpretation $\A$ that satisfies $\phi$ with $|\dom{\A}|<\aleph_{0}$.\footnote{A common definition for the finite model property demands this condition holds for all formulas, but in theory combination quantifier-free formulas are typically used.}
Consider $\N=\mathbb{N}\cup\{\aleph_{0}\}$. A \textbf{minimal model function}~\cite{TinZar05} for $\T$ is a function $\minmod_{\T}:\qf{\Sigma}\rightarrow\N$ such that, if $\phi$ is quantifier-free and $\T$-satisfiable, then $\minmod_{\T}(\phi)=n$ if, and only if: there exists a $\T$-interpretation $\A$ that satisfies $\phi$ with $|\dom{\A}|=n$; and if $\B$ is another $\T$-interpretation that satisfies $\phi$, with $|\dom{\B}|\neq n$, then $|\dom{\B}|>n$.

\begin{example}
    The theory axiomatized by $\{\psi_{\geq 3}\}$ has a computable minimal model function.
    To calculate it on a quantifier-free formula $\phi$, take the cardinality $n$ of the smallest interpretation in equational logic that satisfies $\phi$, which can easily be found algorithmically. 
    If $n<3$, $\minmod(\phi)=3$;
    otherwise $\minmod(\phi)=n$.
\end{example}

%
$\T$ is \textbf{(strongly) polite} if it is smooth and (strongly) finitely witnessable. It is \textbf{shiny} if it is smooth, has the FMP and a computable minimal model function.
%
$\T$ is \textbf{gentle}~\cite{gentle} if for every quantifier-free formula $\phi$, $\spec{\T,\phi}$ is \emph{fully computable}, that is: 
$(i)$ it is computable; 
$(ii)$ it is either co-finite,\footnote{\textit{I.e.}, $\mathbb{N}\setminus \spec{\T,\phi}$ is finite.} or a finite set of finite cardinalities, and there is an algorithm with $\phi$ as input that tells which one is the case;
$(iii)$ if $\spec{\T,\phi}$ is finite, $\max(\spec{\T,\phi})$ is computable, and if it is infinite $\max(\mathbb{N}\setminus \spec{\T,\phi})$ is computable, both with $\phi$ as input.\footnote{If $\spec{\T,\phi}$ or $\mathbb{N}\setminus\spec{\T,\phi}$ are empty, their respective maxima are $0$, as usual, so $\T$ must be decidable as $\max(\spec{\T,\phi})=0$ iff $\phi$ is not $\T$-satisfiable.}

\begin{example}\label{firsteven}
    Consider the $\Sigma_{1}$-theory $\Teven$ (see \cite{CADE}), with axiomatization\\ $\{\neg\psi_{=2n+1} : n\in\mathbb{N}\}$:
    it is not gentle, as $x=x$ has as spectrum the set of even positive numbers, which is neither finite nor cofinite.
\end{example}

\section{Theory combination and natural density}\label{computability and density}
In this section we establish various connections between 
model-theoretic properties of a theory, and its natural density.
We focus our investigation on the empty signature $\Sigma_1$, 
that has a single sort and no function and predicate symbols
other than equality.
Generalizations to non-empty signatures are given
 in \Cref{non empty}.

We start with the empty signature, because 
a theory on such a signature has essentially one natural density, while for the non-empty case we must consider the density with respect to both a formula and the theory (this can also be done on the empty case, but all $\T$-satisfiable formulas will give the same density).
Furthermore, some results will not hold on the non-empty case, such as the third item in \Cref{SI density} below.

\Cref{sec:sufficient} deals with sufficient conditions: if the density satisfies them, then we can deduce some 
combination properties.
\Cref{necessary section} obtains necessary conditions: 
one would use the contrapositive  and conclude that the theory does not have the properties at hand,
and then at least one knows that a different combination method
has to be used.

\subsection{Sufficient conditions}
\label{sec:sufficient}
%
In \Cref{SI density} we identify sufficient conditions for stable infiniteness, the finite model property and finite witnessability, properties that are needed
for Nelson-Oppen combination, shiny combination, and polite combination, respectively.

\begin{restatable}{theorem}{SIdensity}\label{SI density}\label{not FW density}
    If $\T$ is a $\Sigma_{1}$-theory with a well-defined natural density, then the positivity of $\mu(\T)$ is sufficient for $\T$ to:
        1.  be stably infinite;
        2. have the finite model property; and
        3. be finitely witnessable.
    
\end{restatable}

\begin{proof}[sketch]\footnote{Full proofs appear in the appendix.}
%
The proof of the third item is more involved than that of the first two, which is routine. 
Szemer{\'e}di's theorem~\cite{Szemerdi}, which settled a well-know conjecture by Erd{\"o}s and Tur{\'a}n, showed that each set with positive natural density contains arbitrarily long finite subsequences in arithmetic progression (i.e., the difference between two consecutive elements is constant). 
Item $3$ is a similarly flavored result, although with a much simpler proof than that of Erd{\"o}s and Tur{\'a}n, that will guarantee that any theory $\T$ which is not finitely witnessable and  has a well-defined natural density must satisfy $\mu(\T)=0$.

\end{proof}

\begin{example}\label{one ex}
Fix some positive natural number $n$,
    and consider the theory $\Tgeqn$, with axiomatization $\{\psi_{\geq n}\}$.
    It obviously has positive density.
    By \Cref{SI density} it is stably infinite, has the finite model property, and is finitely witnessable.
\end{example}

The following example shows that 
all the reciprocals of \Cref{SI density} are false, a single counterexample being enough for all three. 

\begin{example}\label{example SI density}   
    Take the $\Sigma_{1}$-theory $\Tpt$ with axiomatization $\{\psi_{\geq 2^{n}}\vee\bigvee_{i=0}^{n}\psi_{=2^{i}}: n\in\mathbb{N}\}$, which has interpretations $\A$ with domains whose cardinality is either infinite or a power of two.
    It is stably infinite, has the finite model property and is finitely witnessable,\footnote{A witness being, if $\phi$ has $n$ variables, $\wit(\phi)=\phi\wedge\bigwedge_{i=1}^{2^{n}}x_{i}=x_{i}$, for fresh $x_{i}s$.} 
    but 
    $\mu(\Tpt)=\LIM \frac{|\specn{\Tpt}|}{n}=\LIM  \frac{\lfloor \log_{2}(n)\rfloor+1}{n}=0$.
\end{example}

The following example shows the sharpness of
\Cref{not FW density}, in the sense that its assumption is really needed to reach its conclusions.

\begin{example}\label{first appearance Tbb}
The conclusion of \Cref{not FW density} cannot hold under the assumption that $\mu(\T)=0$.
The theory
$\Tinfty$, with axiomatization $\{\psi_{\geq n} : n\in\No\}$, 
has only infinite models.
It has density $0$ but does not have the finite model property.
The theory $\Tone$, with axiomatization $\{\psi_{=1}\}$, has a single model up to isomorphism, with a single element. 
It has density $0$ but is not stably infinite.
For item $3$, a theory that is not finitely witnessable and has natural density $0$ is $\Tbb$, from \cite{FroCoS}, with axiomatization 
$\{\psi_{\geq \bb(n)}\vee\bigvee_{i=2}^{n}\psi_{=\bb(i)} : n\in\mathbb{N}\setminus\{0,1\}\}$
for $\bb:\mathbb{N}\rightarrow\mathbb{N}$ the busy beaver function
(see \Cref{computable number examples}).
The cardinalities of its finite models are precisely the Busy Beaver numbers,
that is, the elements of the image of $\bb$.
We can show that $\mu(\T)=0$.
In a way, item $3$ of \Cref{not FW density} shows that every theory  not finitely witnessable must, like $\bb$, "escape" all computable functions, and thus have natural density $0$.
\end{example}

\subsection{Necessary conditions}\label{necessary section}
We now move on to the results establishing necessary conditions
for gentleness (\Cref{sec:genneccond}),
smoothness and finite model property (\Cref{smooth section}), 
strong finite witnessability (\Cref{SFW section}), 
the computability of a minimal model function (\Cref{cmmf section}), 
and  finite witnessability (\Cref{FW section}).

\subsubsection{Gentleness}
\label{sec:genneccond}
\begin{restatable}{theorem}{gentlesomething}
\label{gentle zero one}
    If $\T$ is a $\Sigma_{1}$-theory, then $\mu(\T)$  being well-defined and equal to $0$ or $1$ is a necessary condition for $\T$ to be gentle.
\end{restatable}

\begin{proof}[sketch]
By taking a tautology $\phi$ for a gentle $\Sigma_{1}$-theory $\T$, we see that $\spec{\T}=\spec{\T,\phi}$ is either finite (and then its density is $0$) or co-finite (and then its density is $1$. 
\end{proof}

\begin{example}
\label{exwithteven}
Consider the theory $\Teven$ from \Cref{firsteven}:
it's density is $\fr{1}{2}$, what implies by the theorem it is not gentle. 
\end{example}

The reciprocal of \Cref{gentle zero one} is false, as shown by the next
example.
\begin{example}\label{example gentle density}
    \Cref{example SI density} presents a theory $\Tpt$ that has density $0$ but is not gentle (since both $\spec{\Tpt}$ and $\mathbb{N}\setminus\spec{\Tpt}$ are infinite).
    On the other hand, take the $\Sigma_{1}$-theory $\Tnpt$ with axiomatization $\{\neg\psi_{=2^{n}} : n\in\mathbb{N}\}$, which has interpretations $\A$ with either $|\dom{\A}|$ infinite, or $|\dom{\A}|$ finite but not a power of two. 
    It is not gentle, yet
    $\mu(\Tnpt)=\LIM \frac{|\specn{\Tnpt}|}{n}=\LIM  \frac{n-\lfloor \log_{2}(n)\rfloor-1}{n}=1$.
\end{example}

Notice also that both cases of \Cref{gentle zero one} are possible, namely: there are gentle theories with density $0$ and gentle theories with density $1$.
Before showing them, let us present two useful lemmas, that relate gentleness
to other properties:

\begin{restatable}{lemma}{gentlenessimplies}
\label{gentleness implies}
    If $\T$ is gentle, then $\T$ has a computable minimal model function and the finite model property, and therefore is finitely witnessable as well.
\end{restatable}

\begin{restatable}{lemma}{impliesgentleness}
\label{implies gentleness}
    Let $\T$ be a $\Sigma_{1}$-theory: if $\T$ is not stably infinite, or if it is strongly finitely witnessable, then $\T$ is gentle.
\end{restatable}

\begin{example}\label{line eight}~
\begin{enumerate} 
\item The trivial $\Sigma_{1}$-theory $\T_{\geq 1}$, with axiomatization $\{\psi_{\geq 1}\}$, consists of all $\Sigma_1$-interpretations.
It is strongly finitely witnessable (given its axiomatization is given by an universal formula, this is proven in \cite{Sheng2022}), and of course $\specn{\T_{\geq 1}}=[1,n]$ so $\mu(\T_{\geq 1})=1$. 
\item The $\Sigma_{1}$-theory $\Tone$ from \Cref{first appearance Tbb}  is also strongly finitely witnessable and thus gentle (\Cref{implies gentleness}), but $\specn{\Tone}=\{1\}$ so $\mu(\Tone)=0$.
\item An example of a $\Sigma_{1}$-theory that is gentle and has density $0$, but is not strongly finitely witnessable,
is denoted by $\Tmn$, for any fixed $m,n\in\No$.
It has axiomatization $\{\psi_{=m}\vee\psi_{=n}\}$, and its
models have cardinalities $m$ or $n$.
\end{enumerate}
\end{example}

\subsubsection{Smoothness and finite model property}\label{smooth section}

The next result involves both smoothness and the finite model property.

\begin{restatable}{theorem}{smoothdensitytheorem}
\label{smooth density}
    If $\T$ is a $\Sigma_{1}$-theory, $\mu(\T)$ being well-defined and equal to $1$ is necessary  for $\T$ to simultaneously admit smoothness and the finite model property.
\end{restatable}

\begin{proof}[sketch]
 The proof is dual to that of \Cref{SI density}: if a theory is smooth and has the finite model property, it has all sufficiently large numbers as cardinalities of its models, and its density is therefore $1$.
\end{proof}

The following example not only allows one to visualize the use of \Cref{smooth density}, but will also help later in providing examples for each and all possible combination of the properties under consideration.

\begin{example}\label{application of smooth density}
Consider again $\Teven$ from \Cref{exwithteven}, with density $\fr{1}{2}$.
It was already shown in \cite{SZRRBT-21} that $\Teven$ has the finite model property without being smooth, but notice that \Cref{smooth density} perfectly encapsulates an intuition for why that is:
    as the theory has the finite model property, it has a finite model;
    were it smooth, it would have models of all larger cardinalities, and thus density $1$.
\end{example}

\begin{example}\label{example smooth density}
    The reciprocal of \Cref{smooth density} is false, as we can see from the theory $\Tnpt$ defined in \Cref{example gentle density}, which is not smooth. 
\end{example}

Now, \Cref{application of smooth density} shows an example of a theory that has the finite model property but is not smooth. 
But all three other Boolean combinations of these two properties are possible, as seen below.

\begin{example}\label{SM-FMP ex}~\begin{enumerate}
        \item The theory $\T_{\geq 1}$ from \Cref{line eight} is smooth and has the finite model property.
        \item One example of a smooth theory without the finite model property is the $\Sigma_{1}$-theory $\Tinfty$  
        from \Cref{first appearance Tbb}.
        It has density $0$, as it has no finite models.
        \item To see one of a theory that is neither smooth nor has the finite model property, which by \Cref{SI density} must have density $0$, fix an $n\in\No$ and consider the $\Sigma_{1}$-theory $\Tninfty$, defined in \cite{CADE} by the axiomatization $\{\psi_{=n}\vee\psi_{\geq m} : m\in\No\}$. Its finite models must have cardinality $n$.
    \end{enumerate}
\end{example}

\subsubsection{Strong finite witnessability}\label{SFW section}

The following result, which is a corollary of earlier ones, is specially useful: proving a theory is not strongly finitely witnessable is quite challenging;
it involves finding a quantifier-free formula, a set of variables, and an arrangement on that set which fail the conditions to be a strong witness, for every candidate for a strong witness.
Checking whether the theory's density is $0$ or $1$ can be fairly easier.

\begin{restatable}{theorem}{sfwdensitytheorem}
\label{SFW density}
    If $\T$ is a $\Sigma_{1}$-theory, then $\mu(\T)$  being well-defined and equal to $0$ or $1$ is a necessary condition for $\T$ to be strongly finitely witnessable.
\end{restatable}

\begin{proof}
 By \Cref{implies gentleness} and \Cref{gentle zero one}.
\end{proof}

\begin{example}
The theory $\Teven$ from \Cref{exwithteven} is not strongly finitely witnessable, as proven in \cite{SZRRBT-21}, but the proof found there is quite involved, demanding careful use of arrangements.
    Here, we only need to point to the fact that $\Teven$ has natural density $\fr{1}{2}$.
\end{example}

\begin{example}
    The reciprocal of \Cref{SFW density} is false: the theories 
    $\Tpt$ and $\Tnpt$ from
    \Cref{example SI density,example gentle density} have, respectively, densities $0$ and $1$, but neither is strongly finitely witnessable, which follows from the fact that both are stably infinite without being smooth, together with \cite[Theorem~7]{CADE}, which shows
    stably infinite, one-sorted theories that are strongly finitely witnessable are smooth.
\end{example}

\subsubsection{Computability of minimal model functions}\label{cmmf section}

We move now to the question of computability of a minimal model function. 
For this, we first establish in \Cref{cmmf and computability} a connection between this  and the computability of the spectra.

\begin{restatable}{proposition}{cmmfcomputability}\label{cmmf and computability}
    $\T$ is a $\Sigma_{1}$-theory with a computable minimal model function if, and only if, $\spec{\T}$ is computable.
\end{restatable}

This proposition plays an important role in the proof of the theorem below:

\begin{restatable}{theorem}{cmmfdensity}\label{cmmf density}
    If $\T$ is a $\Sigma_{1}$-theory with a well-defined density, the fact that $\mu(\T)$ is a computable number is a necessary condition for $\T$ to have a computable minimal model function.
    Furthermore, for every computable number $0\leq r\leq 1$, there exists a $\Sigma_{1}$-theory $\T$ with $\mu(\T)=r$ that has a computable minimal model function and the finite model property, but is not smooth.
\end{restatable}

The proof of the first part takes a theory $\T$
with a computable minimal model function, and from \Cref{cmmf and computability} one sees that $\spec{\T}$ is computable; we then prove 
that this implies $\mu(\T)$
is itself a computable real number, ruling out non-computable numbers. 
Indeed, if $A$ is a computable set, the sequence $\{|\An|\}_{n\in\mathbb{N}}=\{|\{k\in A : k\leq n\}|\}_{n\in\mathbb{N}}$ is computable (and so is $\{n\}_{n\in\mathbb{N}}$, but that is obvious); if $\mu(A)=r$, we have $r=\LIM\fr{|\An|}{n}$, proving $r$ is computable. 


We prove the second part  by constructing in \Cref{definition of f} below, from two sequences $\{a_{n}\}_{n\in\mathbb{N}}$ and $\{b_{n}\}_{n\in\mathbb{N}}$, a function $f$ whose image (which will also equal the spectrum of the theory to be constructed) will be a computable set and have a density associated to the mediants of the ratios $\fr{a_{n}}{b_{n}}$, where the mediant of the fractions $\fr{a}{b}$ and $\fr{c}{d}$ is the fraction $\fr{(a+c)}{(b+d)}$. Although tedious to prove, it is true that the limit of the mediants of the ratios between two sequences equals the limit of the ration, and this guarantees that the natural density of the image of $f$ will be the limit of $\fr{a_{n}}{b_{n}}$. 

\begin{definition}\label{definition of f}
    Given  sequences $\{a_{n}\}_{n\in\mathbb{N}}$ and $\{b_{n}\}_{n\in\mathbb{N}}$ with $0<a_{n}<b_{n}$ and $a_{n},b_{n}\in\mathbb{N}$, for all $n\in\mathbb{N}$, we define an associated function $f:\No \rightarrow\No $ inductively as follows: $f(n)=n$ for $1\leq n\leq a_{0}$, and $f(n)=a_{0}$ for $a_{0}+1\leq n\leq b_{0}$; and, assuming $f(n)$ defined for all $1\leq n\leq M=\sum_{i=0}^{m}b_{i}$, for any $m\geq 0$, we make $f(n)=n$ for $M+1\leq n\leq M+a_{m+1}$, and $f(n)=M+a_{m+1}$ for $M+a_{m+1}+1\leq n\leq M+b_{m+1}$.
\end{definition}

The construction defined in \Cref{definition of f} is outlined  in \Cref{graphs} $(A)$.
Notice the step-shape of the function, that can be computed by induction on $n$.

\begin{figure}[t]
\centering
\begin{tabular}{cc}
\begin{tikzpicture}[scale=0.2]
\draw[-latex] (0,0) -- (21,0) ; 
\foreach \x in  {0,1,2,3,4,5,6,8,9,10,11,12,13,14,15,16,17,18,19,20} 
\draw[shift={(\x,0)},color=black] (0pt,3pt) -- (0pt,-3pt);
\foreach \x in {0,,,,,5,,,,,10,,,,,15,,,,,20} 
\draw[shift={(\x,0)},color=black] (0pt,0pt) -- (0pt,-3pt) node[below] 
{$\scriptscriptstyle{\x}$};
\draw[-latex] (0,0) -- (0,21) ; 
\foreach \y in  {0,1,2,3,4,5,6,8,9,10,11,12,13,14,15,16,17,18,19,20} 
\draw[shift={(0,\y)},color=black] (3pt,0pt) -- (-3pt,0pt);
\foreach \y in {0,,,,,5,,,,,10,,,,,15,,,,,20} 
\draw[shift={(0,\y)},color=black] (0pt,0pt) -- (-3pt,0pt) node[left] 
{$\scriptscriptstyle{\y}$};
\draw[step=1.0,gray!40,very thin] (0,0) grid (20,20);
\draw[dashed,red] (1,1) -- (20,20);
\draw (1,1) -- (4,4) -- (6,4) -- (7,7) -- (9,9) -- (10,9) -- (11,11) -- (17,17) -- (20,17);
\node at (1,1)[circle,fill,inner sep=1.5pt]{};
\node at (2,2)[circle,fill,inner sep=1.5pt]{};
\node at (3,3)[circle,fill,inner sep=1.5pt]{};
\node at (4,4)[circle,fill,inner sep=1.5pt]{};
\node at (5,4)[circle,fill,inner sep=1.5pt]{};
\node at (6,4)[circle,fill,inner sep=1.5pt]{};
\node at (7,7)[circle,fill,inner sep=1.5pt]{};
\node at (8,8)[circle,fill,inner sep=1.5pt]{};
\node at (9,9)[circle,fill,inner sep=1.5pt]{};
\node at (10,9)[circle,fill,inner sep=1.5pt]{};
\node at (11,11)[circle,fill,inner sep=1.5pt]{};
\node at (12,12)[circle,fill,inner sep=1.5pt]{};
\node at (13,13)[circle,fill,inner sep=1.5pt]{};
\node at (14,14)[circle,fill,inner sep=1.5pt]{};
\node at (15,15)[circle,fill,inner sep=1.5pt]{};
\node at (16,16)[circle,fill,inner sep=1.5pt]{};
\node at (17,17)[circle,fill,inner sep=1.5pt]{};
\node at (18,17)[circle,fill,inner sep=1.5pt]{};
\node at (19,17)[circle,fill,inner sep=1.5pt]{};
\node at (20,17)[circle,fill,inner sep=1.5pt]{};
\end{tikzpicture}
&
\begin{tikzpicture}[scale=0.2]
\draw[-latex] (0,0) -- (21,0) ; 
\foreach \x in  {0,1,2,3,4,5,6,8,9,10,11,12,13,14,15,16,17,18,19,20} 
\draw[shift={(\x,0)},color=black] (0pt,3pt) -- (0pt,-3pt);
\foreach \x in {0,,,,,5,,,,,10,,,,,15,,,,,20} 
\draw[shift={(\x,0)},color=black] (0pt,0pt) -- (0pt,-3pt) node[below] 
{$\scriptscriptstyle{\x}$};
\draw[-latex] (0,-1) -- (0,21) ; 
\foreach \y in  {0,1,2,3,4,5,6,8,9,10,11,12,13,14,15,16,17,18,19,20} 
\draw[shift={(0,\y)},color=black] (3pt,0pt) -- (-3pt,0pt);
\foreach \y in {0,,,,,5,,,,,10,,,,,15,,,,,20} 
\draw[shift={(0,\y)},color=black] (0pt,0pt) -- (-3pt,0pt) node[left] 
{$\scriptscriptstyle{\y}$};
\draw[step=1.0,gray!40,very thin] (0,0) grid (20,20);
\draw[dashed,red] (1,1) -- (20,20);
\draw (1,1) -- (3,1) -- (4,4) -- (5,5) -- (7,5) -- (8,8) -- (10,10) -- (14,10) -- (15,15) -- (20,20);
\node at (1,1)[circle,fill,inner sep=1.5pt]{};
\node at (2,1)[circle,fill,inner sep=1.5pt]{};
\node at (3,1)[circle,fill,inner sep=1.5pt]{};
\node at (4,4)[circle,fill,inner sep=1.5pt]{};
\node at (5,5)[circle,fill,inner sep=1.5pt]{};
\node at (6,5)[circle,fill,inner sep=1.5pt]{};
\node at (7,5)[circle,fill,inner sep=1.5pt]{};
\node at (8,8)[circle,fill,inner sep=1.5pt]{};
\node at (9,9)[circle,fill,inner sep=1.5pt]{};
\node at (10,10)[circle,fill,inner sep=1.5pt]{};
\node at (11,10)[circle,fill,inner sep=1.5pt]{};
\node at (12,10)[circle,fill,inner sep=1.5pt]{};
\node at (13,10)[circle,fill,inner sep=1.5pt]{};
\node at (14,10)[circle,fill,inner sep=1.5pt]{};
\node at (15,15)[circle,fill,inner sep=1.5pt]{};
\node at (16,16)[circle,fill,inner sep=1.5pt]{};
\node at (17,17)[circle,fill,inner sep=1.5pt]{};
\node at (18,18)[circle,fill,inner sep=1.5pt]{};
\node at (19,19)[circle,fill,inner sep=1.5pt]{};
\node at (20,20)[circle,fill,inner sep=1.5pt]{};
\node at (1,1)[circle,fill=red,inner sep=1pt]{};
\node at (2,0)[circle,fill=red,inner sep=1pt]{};
\node at (3,0)[circle,fill=red,inner sep=1pt]{};
\node at (4,1)[circle,fill=red,inner sep=1pt]{};
\node at (5,1)[circle,fill=red,inner sep=1pt]{};
\node at (6,0)[circle,fill=red,inner sep=1pt]{};
\node at (7,0)[circle,fill=red,inner sep=1pt]{};
\node at (8,1)[circle,fill=red,inner sep=1pt]{};
\node at (9,1)[circle,fill=red,inner sep=1pt]{};
\node at (10,1)[circle,fill=red,inner sep=1pt]{};
\node at (11,0)[circle,fill=red,inner sep=1pt]{};
\node at (12,0)[circle,fill=red,inner sep=1pt]{};
\node at (13,0)[circle,fill=red,inner sep=1pt]{};
\node at (14,0)[circle,fill=red,inner sep=1pt]{};
\node at (15,1)[circle,fill=red,inner sep=1pt]{};
\node at (16,1)[circle,fill=red,inner sep=1pt]{};
\node at (17,1)[circle,fill=red,inner sep=1pt]{};
\node at (18,1)[circle,fill=red,inner sep=1pt]{};
\node at (19,1)[circle,fill=red,inner sep=1pt]{};
\node at (20,1)[circle,fill=red,inner sep=1pt]{};
\end{tikzpicture} \\
$(A)$ & $(B)$
\end{tabular}
\caption{$(A)$~Initial values of the function $f$ from \Cref{definition of f}, for $a_{0}=4$, $a_{1}=3$, $a_{2}=7$, $b_{0}=6$, $b_{1}=4$ and $b_{2}=10$.~
$(B)$~Initial values of the functions $\mf$ (in red) and $f$ (in black) from the case with computable $r$ of \Cref{FW density general case}, for $a_{0}=1$, $a_{1}=2$, $a_{2}=3$, $b_{0}=2$, $b_{1}=3$, $b_{2}=5$, $g(1)=1$, $g(2)=1$ and $g(3)=0$.}\label{graphs}
\end{figure}
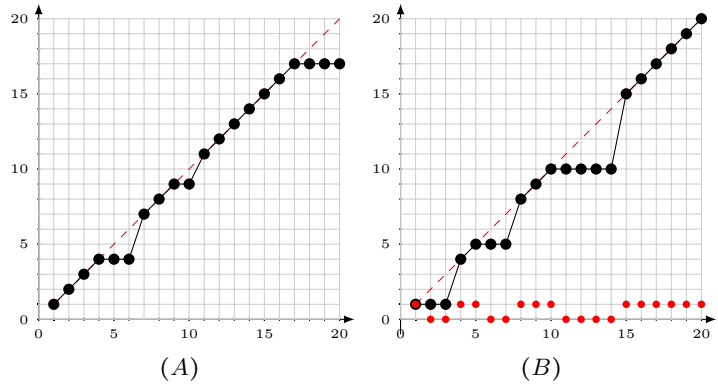

We next show an application of \Cref{cmmf density},
by identifying a theory that according to this theorem
has a computable minimal model function, but this fact seems
difficult to prove without using this theorem.

\begin{example}
    Take $\Omega=0.57824\ldots$ from \Cref{computable number examples}.
    Take the sequence of fractions $\fr{5}{10},\fr{57}{100},\ldots$ converging to $\Omega$, and define a function $f$ as in \Cref{definition of f}:
    so $f(n)=n$ for $1\leq n\leq 4$, and $f(n)=5$ for $5\leq n\leq 10$;
    $f(n)=n$ for $11\leq n\leq 66$, and $f(n)=67$ for $67\leq n\leq 110$,
    and so on.
    Define then a theory $\T_{\Omega}$ with axiomatization $\{\psi_{\geq f(n+1)}\vee\bigvee_{i=1}^{n}\psi_{=f(n)} : n\in \No\}$:
    it has models of size $1$ through $5$, $11$ through $67$, and so on.\footnote{Notice $\T_{\Omega}$ is not the same as the theory from \Cref{computable number examples}: indeed, the construction from \Cref{definition of f} is more general.}
    We can prove that it has natural density $\Omega$, and thus does not have a computable minimal model function.
    It is, however unclear how one would prove this without resorting to \Cref{cmmf density} and \Cref{cmmf and computability}.
\end{example}

The reciprocal of \Cref{cmmf density} is false, in the sense that a theory without a computable minimal model function can still have a computable natural density.

\begin{example}\label{bb}
    Consider again the theory $\Tbb$ from \Cref{first appearance Tbb}, whose natural density we have shown to be the computable number $0$, despite the fact it does not have a computable minimal model function (see \cite[Lemma~128]{LPAR-arXiv}), what by \Cref{cmmf and computability} means $\spec{\Tbb}$ is not computable.
\end{example}

Notice that, from \cite[Lemma~7]{CADE} and \cite[Theorem~4]{LPAR}, the theories in \Cref{cmmf density} are also finitely witnessable without being strongly finitely witnessable. 
They are also not gentle in the case that $0<\mu(\T)<1$.
It is still possible to come up with gentle examples for $\mu(\T)=1$ or $\mu(\T)=0$, as the next example shows.

\begin{example}\label{three ex}
    To obtain $\Sigma_{1}$-theories $\T$, with $\mu(\T)=1$ or $\mu(\T)=0$, that have a computable minimal model function and are gentle but not strongly finitely witnessable it is enough to consider, for the 
    first case, $\T$ with axiomatization $\{\psi_{=1}\vee\psi_{\geq 3}\}$; for the second, $\T$ with axiomatization $\{\psi_{=1}\vee\psi_{=3}\}$.
\end{example}


\subsubsection{Finite witnessability}\label{FW section}

The theorems so far have provided necessary conditions for a theory to be gentle, smooth strongly finitely witnessable, or have a computable minimal model function. 
We now show that this is as far as this goes: namely, we cannot achieve necessary conditions using natural densities for finite witnessability alone.

In fact, any real number $r$ is the natural density of a finitely witnessable theory.
If $r$ is computable then \Cref{cmmf density} already 
constructs a finitely witnessable theory $\T$ with $\mu(\T)=r$,
as \cite[Theorem~4]{LPAR} proved that a theory with a computable minimal model function is finitely witnessable.
However, in the next theorem we construct
such a theory also for non-computable numbers.
In addition, the theorem shows that the generated theory
does not need to have a computable minimal model function even if $r$ is computable.

\begin{restatable}{theorem}{FWdensitygeneralcase}\label{FW density general case}
        If $\T$ is a $\Sigma_{1}$-theory, nothing can be said about $\mu(\T)$ if $\T$ is only known to be finitely witnessable; that is, for every number $0\leq r\leq 1$, there exists a $\Sigma_{1}$-theory $\T$ with $\mu(\T)=r$ that is finitely witnessable yet doesn't have a computable minimal model function.
\end{restatable}

The proof of \Cref{FW density general case} is divided in two cases: 
when $r$ is computable, and when it is not.
When it is not, we write it in decimal notation, 
take the obvious series 
of decimal fractions converging to it,
define the 
function $f$ as in \Cref{definition of f} and take the theory whose spectrum is the image of $f$.

If $r$ is computable, 
which is the difficult case, we take computable sequences $\{a_{n}\}_{n\in\mathbb{N}}$ and $\{b_{n}\}_{n\in\mathbb{N}}$ such that $\fr{a_{n}}{b_{n}}$ converges to $r$, and a non-computable function $g:\No\rightarrow\{0,1\}$.
We then define an auxiliary function $\mf:\No\rightarrow\{0,1\}$ by making, for $M=2\sum_{i=0}^{m}b_{i}$:
$\mf(M+1)=g(m+2)$;
$\mf(n)=0$ for $M+2\leq n\leq 2(b_{m+1}-a_{m+1})+M$;
$\mf(2(b_{m+1}-a_{m+1})+M+1)=0$ if $g(m+2)=1$, and otherwise $\mf(2(b_{m+1}-a_{m+1})+M+1)=1$;
and $\mf(n)=1$ for $2(b_{m+1}-a_{m+1})+M+2\leq n\leq M+2b_{m+1}$.
We make $f(n)=\max\{m\leq n : \mf(m)=1\}$, and take the theory whose spectrum is the image of $f$, axiomatized by $\{\psi_{\geq f(n+1)}\vee\bigvee_{i=1}^{n}\psi_{=f(i)} : n\in\No\}$, which will have a density equal to whatever is the limit of $\fr{a_{n}}{b_{n}}$, i.e. $r$.
An example of this construction appears in \Cref{graphs} $(B)$, where the red dots represent $\mf$, and the black ones represent $f$:
for the corresponding theory we will have $\Spec_{20}(\T)=\{1,4,5,8,9,10,15,16,17,18,19,20\}$.

\subsubsection{Summary}
We can now, as this section about necessary conditions for the empty signature comes to an end, summarize its overall arch: 
we have seen what are all possible values for the density of a theory given some of its theory combination properties, for all such combinations of properties.

We have considered, in total, 7 properties related to theory combination. 
Were we to consider all Boolean combinations of them, we would need to analyze 128 cases;
\cite[Theorems~5,6,7]{LPAR} has shown, however, that for $\Sigma_{1}$ there are only 8 of these possibilities, excluding gentleness.
It may look like we need to analyze 16 possibilities then, but we can cut them down to 9 by using
\Cref{gentleness implies,implies gentleness}.

In \Cref{table},
$\textbf{SI}$ stands for stably infinite;
$\textbf{SM}$ for smooth;
$\textbf{FW}$ for finitely witnessable;
$\textbf{SW}$ for strongly finitely witnessable;
$\textbf{FM}$ for the finite model property;
$\textbf{CF}$ for a computable minimal model function; 
and $\textbf{G}$ for gentle.
$\textbf{REC}$ denotes the set of real computable numbers.

Each line in the table corresponds to a possible combination of properties (that remains possible after \Cref{gentleness implies,implies gentleness}).
For example, the first line corresponds to theories that admit all properties, while the second line correspond to theories that are stably infinite, smooth, have a computable minimal model function, but do not admit any of the other property.

For each possible combination of properties, we list in the table the possible natural densities of theories that admit the corresponding properties.
For example, theories that admit all properties must have density one.\footnote{The theory found in this specific row, $\T_{\geq 1}$ from \Cref{line eight}, is strongly finitely witnessable, and \Cref{implies gentleness} then shows it is also gentle, as implied by the table.}

The column titled "Reference" leads to the result in this paper proving the values are indeed restricted to the mentioned ones;
and the column "Construction" refers to examples of theories having the possible natural densities shown.

\begin{table}[t]
\begin{tabular}{|P{0.6cm}|P{0.6cm}|P{0.6cm}|P{0.6cm}|P{0.6cm}|P{0.6cm}|P{0.6cm}||P{2.8cm}|P{2.0cm}|P{2.0cm}|}
\hline
$\textbf{SI}$ & $\textbf{SM}$ & $\textbf{FW}$ & $\textbf{SW}$ & $\textbf{FM}$ & $\textbf{CF}$ & $\textbf{G}$ & Natural densities & Reference & Construction\\\hline\hline
\multirow{7}{*}{$T$} & \multirow{2}{*}{$T$} & $T$ & $T$ & $T$ & $T$ & $T$ & $1$ & \Cref{smooth density} & \Cref{line eight}\\\hhline{~~--------}
& & $F$ & $F$ & $F$ & $T$ & $F$ & $0$ & \Cref{SI density} & \Cref{SM-FMP ex} \\\hhline{~---------}
& \multirow{5}{*}{$F$} & \multirow{3}{*}{$T$} & \multirow{3}{*}{$F$} & \multirow{2}{*}{$T$} & \multirow{2}{*}{$T$} & $T$ & $\{0,1\}$ & \Cref{gentle zero one} & \Cref{three ex} \\\hhline{~~~~~~----}
& & & & & & $F$ & $\textbf{REC}\cap[0,1]$ & \Cref{cmmf density} & \Cref{cmmf density} \\\hhline{~~~~------}
& & & & $T$ & $F$ & $F$ & $[0,1]$ & \Cref{FW density general case} & \Cref{FW density general case} \\\hhline{~~--------}
& & \multirow{2}{*}{$F$} & \multirow{2}{*}{$F$} & $T$ & $F$ & $F$ & $0$ & \Cref{SI density} & \Cref{first appearance Tbb} \\\hhline{~~~~------}
& & & & $F$ & $T$ & $F$ & $0$ & \Cref{SI density} & \Cref{SM-FMP ex} \\\hline
\multirow{2}{*}{$F$} & \multirow{2}{*}{$F$} & \multirow{2}{*}{$T$} & $T$ & $F$ & $T$ & $T$ & $0$ & \Cref{SI density} & \Cref{line eight} \\\hhline{~~~-------}
 & & & $F$ & $T$ & $T$ & $T$ & $0$ & \Cref{SI density} & \Cref{line eight} \\
\hline
\end{tabular}
\vspace{1em}
\caption{Classification of combinations \textit{vis-{\`a}-vis} their natural densities.}\label{table}
\vspace{-4mm}
\end{table}

\section{Non-empty signatures}\label{non empty}\label{non empty section}

In this section we provide generalizations of the results of \Cref{computability and density} to non-empty signatures.
We are able to do so by considering
$\mu(\T,\phi)$ for all formulas $\phi$, rather than $\mu(\T)$:
this is due to the fact that
in a non-empty signature we can have two quantifier-free satisfiable formulas with distinct densities.

\begin{example}\label{counterexample SFW non empty}
    Take the theory $\T$ on the signature with a unary function $s$, axiomatized by $\{\psi_{=2}\vee\Forall{x}s(x)=x\}$. 
In its models that do not have exactly two elements,
$s$ must be interpreted as the identity. 
    For $\phi_{1}=\neg(s(x)=x)$ and $\phi_{2}=(s(x)=x)$ we have $\mu(\T,\phi_{1})=0$ and $\mu(\T,\phi_{2})=1$.
\end{example}

We start by generalizing items $1$ and $2$ in \Cref{SI density}.
As for the third item of \Cref{SI density},
we show in \Cref{not FW density counterexample} below that it cannot be generalized similarly.

\begin{restatable}{theorem}{SIFMPnonempty}\label{SI FMP non empty}
The positivity of $\mu(\T,\phi)$ for every $\T$-satisfiable quantifier-free formula $\phi$ is   sufficient for $\T$ to 
    be stably infinite and
    have the finite model property.
\end{restatable}

In the next example we show how to use \Cref{SI FMP non empty}.

\begin{example}
    Consider a signature $\Sigma$ with only function symbols, and the $\Sigma$-theory $\T$ of uninterpreted functions.
    For every quantifier-free formula $\phi$ and  $\T$-interpretation $\A$ that satisfies it, we can add an element $a$ to its domain, 
    from that it follows that $\mu(\T,\phi)=1$.
    Using \Cref{SI FMP non empty}, we conclude that the theory of uninterpreted functions is both stably infinite, and has the finite model property.
\end{example}

Next, we generalize the result concerning gentleness to
non empty signatures.
The proof of the following result is, \textit{mutatis mutandis}, the same as \Cref{gentle zero one}.

\begin{theorem}\label{gentle non empty}
    $\mu(\T,\phi)$ being well-defined and equal to $0$ or $1$ for all quantifier-free $\T$-satisfiable formulas $\phi$, is a necessary condition for $\T$ to be gentle.
\end{theorem}

The following theorem generalizes \Cref{smooth density},
and provides a necessary condition for smoothness and the finite model property for non-empty signatures.

\begin{restatable}{theorem}{smoothnonempty}\label{smooth non empty}
Let $\T$ be a theory.
$\mu(\T,\phi)$ being well-defined and equal to $1$ for all $\T$-satisfiable quantifier-free formulas $\phi$ is then a necessary condition for $\T$ to simultaneously be smooth and have the finite model property.
\end{restatable}

\Cref{smooth non empty} can be used to show that the third item of \Cref{SI density}
is not generalizable to non-empty signatures. 

\newcommand{\distinct}[1]{\delta_{{#1}}}
\newcommand\smeq{\scaleobj{0.85}{=}}
\begin{example}\label{not FW density counterexample}
    Consider the function $\bb^{-1}:\mathbb{N}\rightarrow\mathbb{N}$ from \cite{FroCoS}, which is a left inverse of $\bb$, 
    and the theory $\Tbbs$ on the signature with only a single unary function $s$, from the same paper.
    Both $\bb^{-1}$ and $\Tbbs$ are given in \Cref{fig-card-s}.
It is  
 smooth, has the finite model property, but is not finitely witnessable (see \cite[Lemmas~71,72,73]{arxivFroCoS}), meaning $\mu(\Tbbs,\phi)=1$ for all quantifier-free $\Tbbs$-satisfiable formulas $\phi$ by \Cref{smooth non empty}. 
Thus, the obvious generalization of item $3$ of \Cref{not FW density} is not valid.
\end{example}

\Cref{smooth non empty} is also useful to show,
for example, that
a variant of the SMT-LIB theory of bit-vectors
is not smooth.

\begin{example}
Fix $n\in\No$, and
    consider the one-sorted fragment of the SMT-LIB theory $\bv$ of bit-vectors~\cite{BarFT-RR-17} of length $n$, with the usual operations (but without concatenation and extraction). 
    The domain of its interpretations  has cardinality $2^{n}$, and so it  has the finite model property.
    By \Cref{smooth non empty} this theory is not smooth, as for any quantifier-free formula $\phi$ one has $\mu(\bv,\phi)=0$.
\end{example}

\begin{figure}[t]
\begin{mdframed}
\[\psi^{\smeq}_{\geq n}=\Exists{x_{1}}\cdots\Exists{x_{n}}\big[\bigwedge_{i=1}^{n} \NNEQ{x}\wedge[s(x_{i})=x_{i}]\big]\]
\[\psi^{\smeq}_{=n}=\Exists{x_{1}}\cdots\Exists{x_{n}}\big[\bigwedge_{i=1}^{n}\NNEQ{x}\wedge[s(x_{i})=x_{i}]\big]\wedge\Forall{x}[[s(x)=x]\rightarrow\bigvee_{i=1}^{n}x=x_{i}]\big]\]
\[\bb^{-1}(k)=\min\{l : \bb(l+1)>\bb(k)\}\]
\[\ax{\Tbbs}=\{(\psi_{\geq k+1}\wedge\psi^{\smeq}_{\geq \bb^{-1}(k+1)})\vee\bigvee_{i=1}^{k+1}(\psi_{=i}\wedge\psi^{\smeq}_{=\bb^{-1}(i)}) : k\in\mathbb{N}\}\]
\end{mdframed}
\vspace{-2mm}
\caption{The theory $\Tbbs$.}
\vspace{-4mm}
\label{fig-card-s}
\end{figure}

Next, we generalize \Cref{SFW density} to non-empty signatures.

\begin{restatable}{theorem}{SFWnonempty}\label{SFW non empty}
    $\mu(\T,\phi)$ being well-defined and equal to $0$ or $1$ for every quantifier-free $\T$-satisfiable $\phi$  is necessary  for $\T$ to be strongly finitely witnessable.
\end{restatable}

\Cref{counterexample SFW non empty} shows tightness of \Cref{gentle non empty,SFW non empty}: we can have a strongly finitely witnessable,\footnote{
    The strong witness is $\wit(\phi)=\phi\wedge\neg(x=y)$, for fresh variables $x$ and $y$.} gentle\footnote{From the fact it is strongly finitely witnessable and \Cref{implies gentleness}.} theory $\T$ with two quantifier-free $\T$-satisfiable formulas that have densities $0$ and $1$.
It also shows that the positivity in \Cref{SI FMP non empty} cannot hold for only some quantifier-free  $\T$-satisfiable formulas $\phi$, as the theory shown is not stably infinite.
%

The following two theorems generalize, respectively, \Cref{cmmf density,FW density general case}.
For \Cref{cmmf density}, we need an alternative, non-empty version of \Cref{cmmf and computability}.
Iindeed, it is not clear that if the sets $\spec{\T,\phi}$ are all computable, $\T$ should have a computable minimal model function; the reciprocal, however, is true.

\begin{restatable}{proposition}{cmmfcomputabilitygen}\label{cmmf and computability gen}
    If $\T$ is a theory with a computable minimal model function, then $\spec{\T,\phi}$ is computable for all quantifier-free $\T$-satisfiable formulas $\phi$.
\end{restatable}

\begin{restatable}{theorem}{CMMFnonempty}\label{CMMF non empty}
    If $\T$ is a theory with well-defined densities $\mu(\T,\phi)$, for all quan\-ti\-fier-free $\T$-satisfiable formulas $\phi$, the fact that all $\mu(\T,\phi)$ are computable is a necessary condition for $\T$ to have a computable minimal model function.
    Furthermore, for every computable number $0\leq r\leq 1$, there is a theory $\T$ that has a computable minimal model function and a quantifier-free formula $\phi$ with $\mu(\T,\phi)=r$.
\end{restatable}

\begin{restatable}{theorem}{FWnonempty}\label{FW non empty}
    If $\T$ is a theory, and $\phi$ 
    a quantifier-free $\T$-satisfiable formula, nothing can be said about $\mu(\T,\phi)$ if $\T$ is only known to be finitely witnessable; that is, for every computable number $0\leq r\leq 1$, there exists a theory $\T$, that is finitely witnessable, and a quantifier-free formula $\phi$ with $\mu(\T,\phi)=r$.
\end{restatable}

\section{Related work and conclusion}\label{conclusion}

We have studied
connections between
densities and model-theoretic properties. 
\Cref{tabletwo} summarizes our main results.
For each property, we refer to the theorems that characterize 
its possible densities, both for empty and non-empty signatures.

\begin{table}[t]
\centering

\begin{tabular}{|P{4.9cm}||P{3.4cm}|P{3.5cm}|}
\hline
Property & Empty case & Non-empty case\\
\hline\hline
$\textbf{Stable Infiniteness}$ & \Cref{SI density} & \Cref{SI FMP non empty}\\\hline
$\textbf{Finite Model Property}$ & \Cref{SI density} & \Cref{SI FMP non empty}\\\hline
\textbf{Gentleness} & \Cref{gentle zero one} & \Cref{gentle non empty}\\\hline
$\textbf{Smoothness}$ & \Cref{smooth density}& \Cref{smooth non empty}\\\hline
$\textbf{Strong Finite Witnessability}$ & \Cref{SFW density} & \Cref{SFW non empty}\\\hline
$\textbf{Comp. of Min. Mod. Fun.}$ & \Cref{cmmf density} & \Cref{CMMF non empty}\\\hline
$\textbf{Finite Witnessability}$ & \Cref{FW density general case} & \Cref{FW non empty},\Cref{not FW density counterexample}\\\hline
\end{tabular}
\vspace{1em}
\caption{Summary of main results.}\label{tabletwo}
\vspace{-8mm}
\end{table}

We conclude by reviewing related work and sketching 
the next steps.

\subsection{Related work}
\emph{$0$-$1$-laws and densities.}
Studies on spectra and densities go back as far as \cite{Carnap,Glebskii,Fagin}.
While we consider only models of a theory, these results, including the famous $0$-$1$ laws, concern \emph{random models}, that is, any models.
$0$-$1$ laws remain powerful for theories with finite axiomatizations (as we can represent their axiomatizations using a conjunction),
but here we consider also infinite axiomatizations.
Later studies, such as \cite{Compton,Bell} considered densities with respect to a theory, or even a (sufficiently well-behaved) class of models, but have not considered properties associated with theory combination.
We focus on theories, and on the relationship between  their combination properties and the behavior of their density.

\emph{Descriptive complexity.} Note that we use slightly different definitions for the spectrum of a theory than those found in descriptive complexity~\cite{descriptivecomplexity}:
although our definition of $\spec{\T,\phi}$ is the usual one 
for the spectrum of a formula relative to a theory, 
the spectrum of a theory $\T$ is more commonly understood as the map from cardinals to cardinals which, given $\kappa$, returns the number of non-isomorphic models of $\T$ of cardinality $\kappa$.
But for the case of finite cardinalities in the empty signature, there is this map would return either $0$ or $1$.
Then, our definition coincides with taking the pre-image of $1$ in the more standard definition.

\emph{Theory combination properties.} The current paper deals with, among other topics, Boolean combinations of theory combination properties (especially in \Cref{table}), something comprehensively researched in  \cite{CADE,FroCoS,LPAR}.
While those papers study the combinations of properties per se, here we focus on establishing these properties (or lack of) through the analysis of a their density.

\subsection{Future work: many-sorted densities}
\label{rem:manysorted}
In this paper we only considered one-sorted theories, even though many-sorted theories
are commonly used in SMT.
The main reason for that is that densities for many-sorted theories would be defined
on tuples rather than on numbers (i.e. on the cardinalities of the domains rather than
on that of the single domain), and it is unclear how this generalization would materialize.
We leave this investigation for future work, and briefly describe concrete options for such a generalization.

What makes the natural density so natural is the fact that it calculates the ratio of the number of elements in a set $A$ to the number of elements in $\mathbb{N}$ by doing that for numbers under a bound, and then letting said bound go to infinite. 
But there is no single way of doing that in $\mathbb{N}^{m}$, so we are forced to make a choice.
Once fixed a bound $n$, do we, for example:

$(i)$
bound all coordinates simultaneously by $n$ (i.e., $\mu(A)=\LIM\fr{|A_{n}|}{n}$ for $A_{n}=A\cap[n]^{m}$)?
$(ii)$
bound the distance of a tuple to the origin by $n$ (i.e., $A_{n}=A\cap B_{d}(n)$, where $B_{d}(n)=\{\textbf{p}\in\mathbb{N}^{m} : d(\textbf{0},\textbf{p})\leq n\}$, for $\textbf{0}$ the origin)?
$(iii)$
If so, what metric do we use to calculate the distance? Do we use the taxicab distance, where $d_{1}(\textbf{p},\textbf{q})=\sum_{i=1}^{n}|p_{i}-q_{i}|$, or the generalized euclidean distances $d_{m}(\textbf{p},\textbf{q})=(\sum_{i=1}^{n}(p_{i}-q_{i})^{m})^{\fr{1}{m}}$, or something entirely different?
 
There is a plurality of "natural densities" to explore. Even more, while some generalizations will characterize properties w.r.t. the entire set of sorts $\{\s_{1},\ldots,\s_{n}\}$, others will characterize them with respect to some subset of sorts, while others will offer no characterization whatsoever.

All of this is left to a future work, but we expect that the results from the current paper will still be useful for many-sorted logic, 
as many of the potential many-sorted densities would rely on the 
separate projections to each sort. 

\newpage

\bibliography{bib}{}
\bibliographystyle{plain}


\newpage

\appendix

\section{Useful Theorems}
We briefly recall the L{\"o}wenheim-Skolem and compactness theorems (see, e.g.,  \cite{Hodges}).

\begin{theorem}\label{LowenheimSkolem}
    Given a first-order signature $\Sigma$, if a set $\Gamma$ of $\Sigma$-formulas is satisfiable by a $\Sigma$-interpretation with $|\dom{\A}|\geq\aleph_{0}$, then it is satisfied by an interpretation $\B$ with $|\dom{\B}|=\aleph_{0}$.
\end{theorem}

\begin{theorem}\label{compactness}
    Given a first-order signature $\Sigma$, a set $\Gamma$ of $\Sigma$-formulas is satisfiable if, and only if, every finite subset $\Gamma_{0}\subseteq\Gamma$ is satisfiable.
\end{theorem}


\section{\tp{Proof of \Cref{SI density}}{Proof of result \ref{SI density}}}

\begin{lemma}\label{bigger interpretation}
    Let $\T$ be a $\Sigma_{1}$-theory.
    If $\phi$ is a quantifier-free $\Sigma_{1}$-formula, and $\A$ and $\B$ are $\T$-interpretations such that $\B$ satisfies $\phi$, and $|\dom{\A}|\geq|\vars(\phi)^{\B}|$, then there is a $\T$-interpretation $\A^{\prime}$, differing from $\A$ at most on the value assigned to $\vars(\phi)$, that satisfies $\phi$. 
    
    If $\A$ and $\B$ are $\T$-interpretations with $|\dom{\A}|=|\dom{\B}|$, then there is a $\T$-interpretation $\A^{\prime}$, differing from $\A$ at most on the values assigned to variables, such that $\A^{\prime}$ and $\B$ satisfy exactly the same, not necessarily quantifier-free, formulas.
\end{lemma}

\begin{proof}
    Take an injective function $f:\vars(\phi)^{\B}\rightarrow\dom{\A}$ (bijective function $f:\dom{\B}\rightarrow\dom{\A}$ if $|\dom{\A}|=|\dom{\B}|$), and define an interpretation $\A^{\prime}$ such that: $\dom{\A^{\prime}}=\dom{\A}$;  $x^{\A^{\prime}}=f(x^{\B})$ if $x^{\B}$ is in the domain of $f$; and $x^{\A^{\prime}}=x^{\A}$ otherwise. We now prove by structural induction that $\A^{\prime}$ satisfies $\phi$; in the case where $|\dom{\A}|=|\dom{\B}|$ and $f$ is bijective, we must perform the proof simultaneously over all interpretations with the same cardinality as that of $\A$ and $\B$.
    
    \begin{enumerate}
        \item If $\psi$ is an atomic subformula of $\phi$, since our signature is empty it must equal $x=y$, for some variables $x$ and $y$; then $\B$ satisfies $\psi$ iff $x^{\B}=y^{\B}$ and, since $x^{\A^{\prime}}=f(x^{\B})$ and $y^{\A^{\prime}}=f(y^{\B})$, that happens iff $\A^{\prime}$ satisfies $\psi$.
        \item So, suppose $\B$ satisfies the subformulas $\psi$ and $\psi_{i}$ (for $i\in\{1,2\}$) of $\phi$ iff $\A^{\prime}$ satisfies the same formulas.
        \begin{enumerate}
            \item We have that $\B$ satisfies $\neg\psi$ iff it does not satisfy $\psi$, what in turn happens iff $\A^{\prime}$ does not satisfy $\psi$, and so $\B$ satisfies $\neg\psi$ iff $\A^{\prime}$ satisfies $\neg\psi$.
            
            \item Analogously, $\B$ satisfies $\psi_{1}\vee\psi_{2}$ iff it satisfies $\psi_{1}$ or $\psi_{2}$, what happens iff $\A^{\prime}$ satisfies $\psi_{1}$ or $\psi_{2}$, what in turn happens iff $\A^{\prime}$ satisfies $\psi_{1}\vee\psi_{2}$.
            
            \item For the case where $|\dom{\A}|=|\dom{\B}|$, and $\phi$ is not necessarily quantifier-free, suppose $\B$ satisfies $\Exists{x}\psi$, and so there exists an interpretation $\B_{*}$, differing from $\B$ at most on the value assigned to $x$, that satisfies $\psi$. We then define an interpretation $\A^{\prime}_{*}$ differing from $\A^{\prime}$ at most on the value assigned to $x$, where $x^{\A_{*}^{\prime}}=f(x^{\B_{*}})$, and from the strengthened form of the induction hypothesis we have $\A_{*}^{\prime}$ satisfies $\psi$, and thus $\A^{\prime}$ satisfies $\Exists{x}\psi$. 
            
            Reciprocally, if $\A^{\prime}$ satisfies $\Exists{x}\psi$, there is an interpretation $\A^{\prime}_{*}$ differing from $\A^{\prime}$ at most on $x$, and we can then define an interpretation $\B_{*}$, differing of $\B$ at most on $x$ where $x^{\B_{*}}=f^{-1}(x^{\A^{\prime}_{*}})$, so that $\B_{*}$ satisfies $\psi$ and thus $\B$ satisfies $\Exists{x}\psi$.

            \item The cases of $\psi_{1}\wedge\psi_{2}$ and $\psi_{1}\rightarrow\psi_{2}$ (and $\Forall{x}\psi$) can be derived from those, and thus we are done.
            
        \end{enumerate}
    \end{enumerate}   
    That $\A^{\prime}$ is still a $\T$-interpretation follows from the fact that to obtain $\A^{\prime}$ we only (at most) changed the value assigned by $\A$ to some variables.
\end{proof}

\begin{lemma}\label{has infinite model}
    If $\T$ is a $\Sigma_{1}$-theory with an infinite interpretation, it is stably infinite.
\end{lemma}

\begin{proof}
    Let $\A$ be an infinite $\T$-interpretation, $\phi$ a quantifier-free formula, and $\B$ a $\T$-interpretation that satisfies $\phi$: since $|\dom{\A}|\geq|\vars(\phi)^{\B}|$ we can thus apply \Cref{bigger interpretation} to obtain an infinite $\T$-interpretation $\A^{\prime}$ that satisfies $\phi$.
\end{proof}

\begin{lemma}\label{-SI implies bound}
    If $\T$ is a $\Sigma_{1}$-theory that is not stably infinite, then there exists $M\in\mathbb{N}$ such that $|\dom{\A}|\leq M$ for all $\T$-interpretations $\A$.
\end{lemma}

\begin{proof}
    Suppose that, for every $n\in\mathbb{N}$, there exists a $\T$-interpretation $\A_{n}$ with $|\dom{\A_{n}}|\geq n$. If we define $\Gamma=\ax{\T}\cup\{\psi_{\geq m} : m\in\mathbb{N}\}$, for any finite subset $\Gamma_{0}\subseteq\Gamma$, take the maximum $n$ of $m$ such that $\psi_{\geq m}$ occurs in $\Gamma_{0}$.
    Of course, $\A_{n}$ then satisfies $\Gamma_{0}$; it satisfies $\Gamma_{0}\cap\ax{\T}$ as it is a $\T$-interpretation;
    and it satisfies any $\psi_{\geq m}$ in $\Gamma_{0}$, since it satisfies $\psi_{\geq n}$, and $\psi_{\geq n}\rightarrow\psi_{\geq m}$ is valid for all $m\leq n$. From \Cref{compactness}, it follows that $\Gamma$ is satisfiable, meaning there is an infinite $\T$-interpretation $\A$. The result then follows from \Cref{has infinite model}.
\end{proof}

\begin{lemma}\label{-FMP implies bound}
    If $\T$ is a $\Sigma_{1}$-theory that does not have the finite model property, then there exists $M\in\mathbb{N}$ such that $|\dom{\A}|\leq M$ for all finite $\T$-interpretations $\A$.
\end{lemma}

\begin{proof}
    If $\T$ has only finite interpretations it already has the finite model property, so take a quantifier-free formula $\phi$ that is satisfied by no finite $\T$-in\-ter\-pre\-ta\-tions, and an infinite $\T$-interpretation $\A$ that satisfies $\phi$. Suppose that there exists, for each $n\in\mathbb{N}$, a finite $\T$-interpretation $\A_{n}$ with $|\dom{\A_{n}}|\geq n$: then, for $n=|\vars(\phi)^{\A}|$, \Cref{bigger interpretation} guarantees that by changing the value assigned to variables in $\A_{n}$ we can make it satisfy $\phi$, contradicting our hypothesis.
\end{proof}

\SIdensity*

\begin{proof}
    The proofs for the first and second items are the same: assuming $\T$ is not stably infinite, respectively does not have the finite model property, we have from \Cref{-SI implies bound}, respectively \Cref{-FMP implies bound}, that there is an $M\in\mathbb{N}$ such that, for all finite $\T$-interpretations $\A$, $|\dom{\A}|\leq M$. This means that $\specn{\T}\subseteq [1,M]$ for all $n\geq 1$, and so $|\specn{\T}|\leq M$ in that case: therefore, $\mu(\T)=\LIM \fr{|\specn{\T}|}{n}\leq \LIM \fr{M}{n}=0$. Thus $\mu(\T)=0$, and by the contrapositive we have the desired result.

    For the third item, from the hypothesis in the theorem's statement, $\mu(\T)$ is well-defined; suppose, in addition, that $\mu(\T)>0$. 
    Let $\{a_{n}\}_{n\in\mathbb{N}}$ be the increasing sequence of the elements of $\spec{\T}$ (if this set were to be finite, it would be bounded by some $M\in\mathbb{N}$, and we would have $\mu(\T)\leq \LIM\fr{M}{n}=0$). There are then two cases to consider: $(1)$ there is a computable sequence $\{b_{n}\}_{n\in\mathbb{N}}$ such that $a_{n}\leq b_{n}$ for all $n\in\mathbb{N}$; and $(2)$ for all computable sequences $\{b_{n}\}_{n\in\mathbb{N}}$ and all $N_{0}\in\mathbb{N}$, there is an $n\geq N_{0}$ such that $a_{n}> b_{n}$. Notice that the negation of $(2)$ implies $(1)$, meaning these are indeed the only cases we need to consider: if there is a computable sequence $\{b_{n}\}_{n\in\mathbb{N}}$ and an integer $N_{0}$ such that $b_{n}\geq a_{n}$ for all $n\geq N_{0}$, we define the computable sequence $\{c_{n}\}_{n\in\mathbb{N}}$ by making $c_{n}=b_{N_{0}}$ for $n\leq N_{0}$, and $c_{n}=b_{n}$ for $n\geq N_{0}$, which satisfies the conditions in case $(1)$. We will prove that case $(1)$ implies $\T$ is finitely witnessable; and that $(2)$ implies $\mu(\T)=0$, contradicting our assumption.

    \begin{enumerate}
        \item Let $x_{i}$ be fresh variables, and define $N(\phi)$, for a quantifier-free formula $\phi$, as $b_{n}$, for $n=|\vars(\phi)|$. We define $\wit(\phi)=\phi\wedge\bigwedge_{i=1}^{N(\phi)}x_{i}=x_{i}$: this is obviously a function from quantifier-free formulas into themselves, and is computable given $b_{n}$, and thus $N(\phi)$, are computable. Furthermore, $\phi$ and $\Exists{\overarrow{x}}\wit(\phi)$ are clearly $\T$-equivalent, given $\wit(\phi)$ is the conjunction of $\phi$ and a tautology and thus already equivalent to $\phi$.

        Finally, let $\phi$ be a quantifier-free formula, and $\A$ a $\T$-interpretation that satisfies $\wit(\phi)$ (and thus $\phi$, since the two are equivalent). 
        Take again $n=|\vars(\phi)|$, a set $A$ with $a_{n}-|\vars(\phi)^{\A}|$ elements (a non-negative quantity, since $\{a_{n}\}_{n\in\mathbb{N}}$ is increasing and therefore $a_{n}\geq n=|\vars(\phi)^{\A}|$), disjoint from $\dom{\A}$, and define an interpretation $\B$ by making: 
        $\dom{\B}=\vars(\phi)^{\A}\cup A$ (so $|\dom{\B}|=a_{n}$, making $\B$ a $\T$-interpretation, as there is only one $\T$-interpretation of a given size up to isomorphism by \Cref{bigger interpretation}); $x^{\B}=x^{\A}$ for all $x\in\vars(\phi)$ (so $\B$ satisfies $\phi$); 
        \[x_{i}\in \{x_{i} : 1\leq i\leq N(\phi)\}\mapsto x_{i}^{\B}\in\dom{\B}\]
        a surjective map (what is possible, as $N(\phi)=b_{n}\geq a_{n}$); and $x^{\B}$ defined arbitrarily for all other variables. 
        Because $\{x_{i} : 1\leq i\leq N(\phi)\}$ is contained in $\vars(\wit(\phi))$, and $\{x_{i} : 1\leq i\leq N(\phi)\}^{\B}=\dom{\B}$, we have $\dom{\B}=\vars(\wit(\phi))^{\B}$ (since obviously the former set at least contains the latter): it then follows that $\wit$ is a witness.

        \item Consider the computable sequence $b_{n}=n^{2}$: take $m_{0}$ as the minimum of the positive integers $m$ such that $a_{m}> b_{m}$ and, inductively, $m_{n+1}$ as the minimum of the integers $m>m_{n}$ such that $a_{m}> b_{m}$; this way, $a_{m_{n}}>b_{m_{n}}=m_{n}^{2}$ for all $n\in\mathbb{N}$. Of course, for all $n\in\mathbb{N}$ one finds $|\Spec_{a_{m_{n}}}(\T)|=m_{n}$, and $|[1,a_{m_{n}}]|=a_{m_{n}}> m_{n}^{2}$, thus $\fr{|\Spec_{a_{m_{n}}}(\T)|}{|[1,a_{m_{n}}]|}< \fr{1}{m_{n}}$. This means that a subsequence of $\fr{|\Spec_{n}(\T)|}{|[1,n]|}$ (which converges to $\mu(\T)$, as this value was assumed in the theorem's statement to be well-defined) converges to $0$, meaning $\mu(\T)=0$.
        
    \end{enumerate}

\end{proof}



    



    

\section{\tp{Proof of \Cref{gentle zero one}}{Proof of result \ref{gentle zero one}}}

\gentlesomething*
\begin{proof}
    If $\spec{\T}$ is finite, there is an $M$ such that $|\specn{\T}|\leq |\spec{\T}|\leq M$, and therefore $\mu(\T)=\LIM\fr{M}{n}=0$.
    If $\spec{\T}$ is co-finite instead, there is an $M$ such that for all $n>M$ one has $|\specn{\T}|\geq n-M$, and so $\mu(\T)\geq\LIM\fr{(n-M)}{n}=1$.
\end{proof}

\section{\tp{Proof of \Cref{gentleness implies}}{Proof of result \ref{gentleness implies}}}

\gentlenessimplies*

\begin{proof}
Suppose $\T$ is gentle.
From the definition of gentleness we get there must exist a finite number in the spectrum of each $\T$-satisfiable quantifier-free formula $\phi$;
this means there is a finite $\T$-interpretation that satisfies $\phi$, so $\T$ indeed has the finite model property.

Given a quantifier-free formula $\phi$, we can decide whether $\spec{\T,\phi}$ is finite, and then computably calculate $\max(\spec{\T,\phi})$, or co-finite, and in this case calculate $\max(\mathbb{N}\setminus\spec{\T,\phi})$.
Using that $\spec{\T,\phi}$ is computable, we can also algorithmically obtain $\min(\spec{\T,\phi})$: 
if $\spec{\T,\phi}$ is finite, we test which $n\leq \max(\spec{\T,\phi})$ are in $\spec{\T,\phi}$ and take their minimum (if $\max(\spec{\T,\phi})$ is $0$, we just set $\min(\spec{\T,\phi})$ to $0$ as well, for simplicity);
if $\spec{\T,\phi}$ is co-finite, we test which $n\leq \max(\mathbb{N}\setminus\spec{\T,\phi})+1$ are in $\spec{\T,\phi}$, and again take their minimum.    

We state that $\minmod_{\T}(\phi)=\min(\spec{\T,\phi})$ is a computable minimal model function, being certainly computable.
Assume then $\phi$ is $\T$-satisfiable, and so $\min(\spec{\T,\phi})>0$.
For the first direction, suppose there is a $\T$-interpretation $\A$ that satisfies $\phi$ with $|\dom{\A}|<\min(\spec{\T,\phi})$;
since $|\dom{\A}|$ is in $\spec{\T,\phi}$, this contradicts the fact $\min(\spec{\T,\phi})$ is the minimum element of that set.
Now, since $\min(\spec{\T,\phi})$ is in $\spec{\T,\phi}$, there is a $\T$-interpretation $\A$ that satisfies $\phi$ with $|\dom{\A}|=\min(\spec{\T,\phi})$, and thus we have proved $\min(\spec{\T,\phi})$ is indeed a minimal model function.

Finite witnessability follows from \cite[Theorem~4]{LPAR-arXiv}.
\end{proof}

\section{\tp{Proof of \Cref{implies gentleness}}{Proof of result \ref{implies gentleness}}}

\begin{lemma}
\label{sigma one decidable}
    A $\Sigma_{1}$-theory is always decidable.
\end{lemma}

\begin{proof}
    Take a $\Sigma_{1}$-theory $\T$.
    We divide the proof in two cases:
    for each of them we find an algorithm to decide whether a quantifier-free formula $\phi$ is $\T$-satisfiable;
    notice, however, that there is no decision method with $\T$ as input that returns which of the described algorithms is the correct one.
    \begin{enumerate}
        \item Suppose $\T$ has a maximum model $\A$, with $|\dom{\A}|=M$;
        we state that $\phi$ is $\T$-satisfiable if and only if there is an interpretation $\B$ in equational logic that satisfies $\phi$ with $|\dom{\B}|\leq M$, meaning we are done since equational logic is decidable.
        Indeed, if $\B$ is an interpretation in equational logic that satisfies $\phi$ with $|\dom{\B}|\leq M$, we can change the values assigned to the variables in $\phi$ by $\A$ in order to obtain a $\T$-interpretation $\A^{\prime}$ that satisfies $\phi$ (see \Cref{bigger interpretation}).
        And, if $\C$ is a $\T$-interpretation that satisfies $\phi$, given the facts that $|\dom{\C}|\leq M$ and $\C$ is an interpretation in equational logic as well, we are done.

        \item If $\T$ doesn't have a maximum model, we state that all quantifier-free formulas satisfiable in equational logic are also $\T$-satisfiable.
        Indeed, if $\phi$ is satisfied by the interpretation $\B$ of equational logic, we know there is a $\T$-interpretation $\A$ with $|\dom{\A}|\geq \B$, and again by \Cref{bigger interpretation} we are done.
    \end{enumerate}
\end{proof}

\impliesgentleness*

\begin{proof}
    If $\T$ is not stably infinite, by \Cref{-SI implies bound} $\T$ has a maximum interpretation $\A$, say with $|\dom{\A}|=M$.
    $\spec{\T,\phi}$ is then always finite (see \Cref{bigger interpretation});
    $\max(\spec{\T,\phi})=M$ if $\phi$ is $\T$-satisfiable (we use \Cref{sigma one decidable}), and $0$ otherwise;
    and a cardinality is in $\spec{\T,\phi}$ if it is both in $\spec{\T}$ and larger than or equal to the minimum cardinality of an interpretation in equational logic that satisfies $\phi$ ($\spec{\T}$ is a finite list, so it can be hardcoded).
    In summary, $\T$ is gentle.
    
    If $\T$ is strongly finitely witnessable but not stably infinite the previous reasoning already implies $\T$ is gentle, so assume $\T$ is strongly finitely witnessable and stably infinite.
    $\T$ is then also smooth by \cite[Theorem~7]{CADE}, and so $\spec{\T}$ is co-finite:
    let $M$ be the maximum of $\mathbb{N}\setminus\spec{\T}$. 
    $\spec{\T,\phi}$ is then always co-finite, by \Cref{bigger interpretation} and the fact that strong finite witnessability implies the finite model property;     $\max(\mathbb{N}\setminus\spec{\T,\phi})=0$ if $\phi$ is not $\T$-satisfiable (what can be decided from \Cref{sigma one decidable}), and otherwise we make $N$ the cardinality of the smallest interpretation in equational logic that satisfies $\phi$, and then $\max(\mathbb{N}\setminus\spec{\T,\phi})=\max\{M,N-1\}$, both computable quantities;
    and finally, $n\in\spec{\T,\phi}$ iff $n>\max\{M,N-1\}$, meaning $\spec{\T,\phi}$ is computable.
\end{proof}

\section{\tp{Proof of \Cref{smooth density}}{Proof of result \ref{smooth density}}}

\smoothdensitytheorem*

\begin{proof}
    Let $\phi$ be any tautology, and from the fact that $\T$ has the finite model property there exists a finite $\T$-interpretation $\B$ that satisfies $\phi$: 
    from the fact $\T$ is smooth it follows that for all $n\geq M=|\dom{\B}|$ there exists a $\T$-interpretation $\A_{n}$ with $|\dom{\A_{n}}|=n$, so $|\specn{\T}|=|\Spec_{M}(\T)|+(n-M)$ for all $n\geq M$. 
    Thus we have $\mu(\T)=\LIM \fr{|\specn{\T}|}{n}\geq \lim_{m\rightarrow\infty} \fr{m}{(m+M)}=1$; 
    therefore $\mu(\T)=1$, as we wished to prove.
\end{proof}

\section{\tp{Proof of \Cref{SFW density}}{Proof of result \ref{SFW density}}}

\sfwdensitytheorem*

\begin{proof}
    From \cite[Theorem~2]{FroCoS} we know strong finite witnessability implies the final model property. From \cite[Theorem~7]{CADE}, we know that if $\T$ is stably infinite, then it is smooth, and so $\mu(\T)=1$ from \Cref{smooth density}; if $\T$ is not stably infinite, \Cref{SI density} guarantees that $\mu(\T)=0$.
\end{proof}

\section{\tp{Proof of \Cref{cmmf and computability}}{Proof of result \ref{cmmf and computability}}}

\cmmfcomputability*

\begin{proof}
    For the right-to-left direction,
    suppose the spectrum $\spec{\T}$ is computable: if $\T$ does not have the finite model property, by \Cref{-FMP implies bound} one gets 
    that there is a natural number $M$ that bounds the cardinalities of all finite models of $\T$. Thus, a minimal model function can simply be obtained by checking for all models up to $M$ whether they satisfy the input formula, and stopping at the first one (recall that since the signature is empty a model is uniquely determined by its cardinality).
    If $\T$ is empty (has no models) then any computable
    function is a minimal model function, as no formula
    is $\T$-satisfiable. 
    
    We may then assume that $\T$ both has the finite model property and is not contradictory (so $\spec{\T}\neq\emptyset$) and take, if $\phi$ is not a contradiction (what can be determined algorithmically in equational logic): the set of variables $V$ in $\phi$; the set $\eq{V}$ of all the equivalences $E$ on $V$; $M=\sup(\spec{\T})$ (which equals $\aleph_{0}$ if $\spec{\T}$ is infinite, and is otherwise always a positive
    natural number since $\spec{\T}\neq\emptyset$),
    $M(\phi)=\min\{|V/E| : \text{$\delta_{V}^{E}$ implies $\phi$}\}$ (a quantity that is computable since the problem may be reduced to one in equational logic), and we state that
    \[
  \minmod_{\T}(\phi) =
  \begin{cases}
    \aleph_{0}           & \begin{aligned}
       & \text{if $\phi$ is contradictory,} \\
       & \text{or $M(\phi)>M$;}
    \end{aligned} \\
    \min\{M(\phi)\leq k\leq M : k\in \spec{\T}\} & \text{otherwise}
  \end{cases}
\]
    is a computable minimal model function for $\T$. 
    Indeed, if $\phi$ is a contradiction, what the decision procedure for equational logic can tell us, $\minmod_{\T}(\phi)=\aleph_{0}$; 
    if $\phi$ is not a contradiction, but it is a $\T$-contradiction, then $M(\phi)>M$ (indeed, if $M\geq M(\phi)$ there is a $\T$-interpretation $\A$ with $|\dom{\A}|=M$, by appealing to \Cref{compactness} if necessary; and by \Cref{bigger interpretation} we get that $\phi$ is $\T$-satisfiable), and again $\minmod_{\T}(\phi)=\aleph_{0}$. Finally, in the case that $\phi$ is $\T$-satisfiable, since $\T$ has the finite model property by assumption, there is a finite $\T$-interpretation $\A$ that satisfies $\phi$ and, of course, $|\dom{\A}|\geq M(\phi)$ (because otherwise $\A$ induces an equivalence $\delta_{V}^{E}$ that implies $\phi$ with $|V/E|<M(\phi)$): therefore the set $\{M(\phi)\leq k\leq M : k\in \spec{\T}\}$ is not empty (it contains at least $|\dom{\A}|$), and finding its minimum is easy.
    
    So assume $\phi$ is $\T$-satisfiable, and take: an $E\in\eq{V}$ with $|V/E|=M(\phi)$ (and necessarily $M(\phi)<\aleph_{0}$); the minimum $n$ of $M(\phi)\leq k\leq M$ such that $k\in \spec{\T}$;
    and a set $X$ with $|X|=n-M(\phi)$ disjoint from $V/E$. We then create an interpretation $\B$ with: $\dom{\B}=(V/E)\cup X$; $x^{\B}=[x]$ for all $x\in V$, where $[x]$ is the equivalence class under $E$ represented by $x$; and $x^{\B}$ defined arbitrarily for variables $x$ not in $V$. Because $|\dom{\B}|=M(\phi)+(n-M(\phi))=n$, we know thanks to \Cref{bigger interpretation} that $\B$ is a $\T$-interpretation; furthermore, $\B$ satisfies $\delta_{V}^{E}$ by definition, thus satisfying $\phi$, so that we indeed have a $\T$-interpretation that satisfies $\phi$ with cardinality $n$. 
    
    Assume now that there is a $\T$-interpretation $\C$ that satisfies $\phi$ with $|\dom{\C}|=m<n$, and let $F$ be the equivalence induced by $\C$ on $V$, implying that $\C$ satisfies $\delta_{V}^{F}$. We then have that $M(\phi)\leq |V/F|\leq m\leq M$, and that $m\in \spec{\T}$, despite the fact that $m<n$ and $n$ should be the minimum such element. So $\T$ indeed has a computable minimal model function.

    For the reciprocal, the left-to-right direction, suppose that $\T$ has a computable minimal model function $\minmod_{\T}$, and we shall consider two cases: one where $\spec{\T}$ is finite, and one where it is infinite. 
    If $\spec{\T}$ is finite we have nothing left to do, as it is enough to simply hardcode these values into an algorithm that decides whether an element is in $\spec{\T}$. 
    If $\spec{\T}$ is instead infinite, the formulas $\NNEQ{x}$ are $\T$-satisfiable for all $n\in\mathbb{N}$: indeed, since $\spec{\T}$ is infinite, we can always find $m\in\spec{\T}$ such that $m\geq n$, and thus there exists a $\T$-interpretation $\A$ with $|\dom{\A}|=m$; 
    by changing at most the values assigned to the variables $x_{1},\ldots,x_{n}$ (assumed fresh), we get a $\T$-interpretation that satisfies $\NNEQ{x}$. 
    We then define $f(0)=\minmod_{\T}(x=x)$ and, assuming $f(m)$ defined, $f(m+1)=\minmod_{\T}(\NNNEQ{x}{f(m)+1})$;
    we state that $n\in\spec{\T}$ iff $n\in\{f(0),\ldots,f(n)\}$. 
    
    That this results in a decision procedure follows from the fact that $\minmod_{\T}$ is assumed to be computable and so is producing the formulas $\NNNEQ{x}{f(m)+1}$; 
    we have left to prove that it is both sound and complete. 
    If $n\in\{f(0),\ldots,f(n)\}$, there exists an $0\leq m< n$ such that $n=\minmod_{\T}(\NNNEQ{x}{f(m)+1})$, and so there is a $\T$-interpretation $\A$ that satisfies $\NNNEQ{x}{f(m)+1}$ with $|\dom{\A}|=n$; 
    of course $\A$ is a $\T$-interpretation with $|\dom{\A}|=n$, so $n\in\spec{\T}$ and indeed the algorithm is sound. 
    Reciprocally, if $n\in\spec{\T}$, suppose that $f(m)<n<f(m+1)$ for some $0\leq m<n$ (there is either such an element or $n=f(m)$ and we have nothing to prove, since $f(m+1)\geq f(m)+1$): because $n$ is in $\spec{\T}$, there is a $\T$-interpretation $\A$ with $|\dom{\A}|=n$; but since $n>f(m)$, $n\geq f(m)+1$, and so $\A$ satisfies $\NNNEQ{x}{f(m)+1}$. This, in combination with the fact that $f(m+1)>n$, contradicts that $f(m+1)=\minmod_{\T}(\NNNEQ{x}{f(m)+1})$, proving that the algorithm is also complete and finishing the proof.
\end{proof}

\section{\tp{Proof of \Cref{cmmf density}}{Proof of result \ref{cmmf density}}}

\begin{lemma}\label{mediant}
    Given a sequence $\fr{a_{n}}{b_{n}}$ converging to a $0<r<1$, the sequence of partial mediants
    \[\left\{\cdots,\frac{a_{0}+\cdots+a_{n}}{b_{0}+\cdots+b_{n}},\frac{a_{0}+\cdots+a_{n}+1}{b_{0}+\cdots+b_{n}+1},\cdots,\frac{a_{0}+\cdots+a_{n}+a_{n+1}}{b_{0}+\cdots+b_{n}+a_{n+1}}, \frac{a_{0}+\cdots+a_{n}+a_{n+1}}{b_{0}+\cdots+b_{n}+a_{n+1}+1},\cdots\right\}\]
    also converges to $r$.
\end{lemma}

\begin{proof}
    Because $0<r<1$, we can assume that $0<a_{n}<b_{n}$ for all $n\in\mathbb{N}$, what makes the definition of the partial mediants consistent. If we write $A(n)=\sum_{i=0}^{n}a_{n}$, and analogously for $B(n)$, and $A(m,n)=\sum_{i=n+1}^{m}a_{i}$ for $m>n$, and the same for $B(m,n)$, we see that the elements in our sequence of index $B(n)\leq i\leq B(n)+a_{n+1}$ lie between $\fr{A(n)}{B(n)}$ and $\fr{A(n+1)}{(B(n)+a_{n+1})}$, while those with index $B(n)+a_{n+1}\leq i\leq B(n+1)$ lie between $\fr{A(n+1)}{(B(n)+a_{n+1})}$ and $\fr{A(n+1)}{B(n+1)}$. So, it is enough to prove that both $\fr{A(n)}{B(n)}$ and $\fr{A(n+1)}{(B(n)+a_{n+1})}$ converge to $r$ in order to sandwich the whole sequence of partial mediants into converging to $r$.

    Take an $\epsilon>0$: because $\LIM a_{n}/b_{n}=r$, there exists an $n_{0}\in\mathbb{N}$ such that $|r-(\fr{a_{n}}{b_{n}})|\leq\fr{\epsilon}{2}$ (and thus $|rb_{n}-a_{n}|\leq b_{n}\fr{\epsilon}{2}$) for all $n\geq n_{0}$. Then, for $n>m_{0}=\max\{n_{0},k\}$, where $k=\lceil \fr{2|rB(n_{0})-A(n_{0})|}{\epsilon}\rceil$, we have
    \[\left|r-\frac{A(n)}{B(n)}\right|=\left|r-\frac{A(n_{0})+A(n,n_{0})}{B(n_{0})+B(n,n_{0})}\right|=\]
    \[\left|\frac{rB(n_{0})+rB(n,n_{0})-A(n_{0})-A(n,n_{0})}{B(n_{0})+B(n,n_{0})}\right|\leq\]
    \[\left|\frac{rB(n_{0})-A(n_{0})}{B(n_{0})+B(n,n_{0})}\right|+\left|\frac{rB(n,n_{0})-A(n,n_{0})}{B(n_{0})+B(n,n_{0})}\right|=\]
    \[\left|\frac{rB(n_{0})-A(n_{0})}{B(n)}\right|+\left|\frac{\sum_{i=n_{0}+1}^{n}rb_{i}-a_{i}}{B(n)}\right|\leq\]
    \[\frac{1}{n}|rB(n_{0})-A(n_{0})|+\frac{1}{B(n)}\sum_{i=n_{0}+1}^{n}|rb_{i}-a_{i}|\leq\]
    \[\frac{\epsilon}{2}+\frac{1}{B(n)}\sum_{i=n_{0}+1}^{n}\frac{b_{i}\epsilon}{2}=\frac{\epsilon}{2}+\frac{B(n,n_{0})\epsilon}{2B(n)}\leq\frac{\epsilon}{2}+\frac{\epsilon}{2}=\epsilon,\]
    thus proving $\LIM \fr{A(n)}{B(n)}=r$.\footnote{Notice we have used that $B(n)\geq n$.} The proof for the sequence $\fr{A(n+1)}{(B(n)+a_{n+1})}$ is similar, so we only highlight the differences. $\fr{\epsilon}{2}$ must be replaced by $\fr{\epsilon}{3}$ since we get an extra term $\fr{a_{n+1}(1-r)}{(B(n)+a_{n+1})}$, which again can be bounded from the fact that $B(n)+a_{n+1}\geq n+1$: we just need to choose $m_{0}=\max\{n_{0},K\}$, where $K=\lceil \fr{3\max\{|rB(n_{0})-A(n_{0})|, a_{n+1}(1-r)\}}{\epsilon}\rceil$.
\end{proof}

\begin{lemma}\label{comp. set comp. number}
    If $A$ is a computable set with a well-defined density, $\mu(A)$ is a computable number.
\end{lemma}

\begin{proof}
    Since we have an algorithm for deciding whether a non-negative integer is in $A$ or not, we also have one for the sequence $a_{n}$
    (where the input is $n$ and the output is $a_n$),
    such that $a_{0}=1$ if $1\in A$, and $a_{0}=0$ otherwise; and $a_{n}=a_{n-1}+1$ if $n+1\in A$, and $a_{n}=a_{n-1}$ otherwise. Of course, the sequence $b_{n}=n+1$ is also computable. Now, by definition of $a_{n}$, $a_{n}=|\AA{n+1}|$, and therefore $\mu(A)=\LIM \fr{|\An|}{n}=\LIM \fr{a_{n-1}}{b_{n-1}}$, so $\mu(A)$ is indeed computable.
\end{proof}

\cmmfdensity*

\begin{proof}
    The first part of the theorem follows from \Cref{cmmf and computability,comp. set comp. number}.

    If $r=0$, we take the theory $\Tpt$ from \Cref{example SI density}, axiomatized by $\{\psi_{\geq 2^{n+1}}\vee\bigvee_{i=1}^{n}\psi_{=2^{n}} : n\in\No\}$; 
    if $r=1$, we take the theory $\Tnpt$ from \Cref{example gentle density}, axiomatized by $\{\neg\psi_{=2^{n}} : n\in\No\}$. 
    We already know $\mu(\T_{0})=0$ and $\mu(\T_{1})=1$:
    since their spectra are, respectively, the sets of power of two and non powers of two, both computable, we have by \Cref{cmmf and computability} that both have computable minimal model functions.

    With these cases out of the way, we can assume that $0<r<1$, and thus there exist computable sequences (since $r$ is computable) $\{a_{n}\}_{n\in\mathbb{N}}$ in $\mathbb{N}$, and $\{b_{n}\}_{n\in\mathbb{N}}$ in $\No $ such that $r=\LIM \fr{a_{n}}{b_{n}}$ with $0<a_{n}<b_{n}$ for all $n\in\mathbb{N}$. 
    We then take the function $f$ associated with $\{a_{n}\}_{n\in\mathbb{N}}$ and $\{b_{n}\}_{n\in\mathbb{N}}$, from \Cref{definition of f}, and define $\T$ to be axiomatized by $\{\psi_{\geq f(n+1)}\vee\bigvee_{i=1}^{n}\psi_{=f(i)} : n\in\No \}$, from what it is clear that $\spec{\T}=\{f(n) : n\in\No \}$. Since $\{a_{n}\}_{n\in\mathbb{N}}$ and $\{b_{n}\}_{n\in\mathbb{N}}$ are computable, so is $f$: 
    indeed, the computability of $\{b_{n}\}_{n\in\mathbb{N}}$ implies $M=\sum_{i=0}^{m}b_{i}$ is computable;
    and then the computability of $\{a_{n}\}_{n\in\mathbb{N}}$ implies $f(n)=n$, for $M+1\leq n\leq M+a_{m}$, and $f(n)=n+a_{m}$, for $M+a_{m}+1\leq n\leq M+b_{m}$, are computable. 
    Therefore $\spec{\T}$ is computable: 
    indeed, as $\spec{\T}$ is the image of the non-decreasing function $f$, and $f(1+\sum_{i=0}^{m}b_{i})=1+\sum_{i=0}^{m}b_{i}$, it is enough, to test whether $n\in\spec{\T}$, to check if $n\in\{f(1),\ldots,f(1+\sum_{i=0}^{n}b_{i})\}$ (recall $\{b_{n}\}_{n\in\mathbb{N}}$ is a computable sequence of positive numbers, so $\sum_{i=0}^{n}b_{i}\geq n$). 
    From \Cref{cmmf and computability} we have $\T$ has a computable minimal model function, and thus from \cite[Theorem~4]{LPAR} we have $\T$ is finitely witnessable.
    Finally, $\mu(\T)$ will equal the limit of partial mediants as defined in \Cref{mediant}, if the latter is well-defined, what the result further shows it is and equals in turn $r$: indeed, assuming for an inductive argument that $\fr{|\Spec_{M}{\T}|}{M}$ equals $\fr{(a_{0}+\cdots+a_{m})}{(b_{0}+\cdots+b_{m})}$ for $M=b_{0}+\cdots+b_{m}$, since $f(n)=n$ for $M+1\leq n\leq M+a_{m+1}$ we obtain 
    \[\frac{|\specn{\T}|}{n}=\frac{a_{0}+\cdots+a_{m}+(n-M)}{n}=\frac{a_{0}+\cdots+a_{m}+(n-M)}{b_{0}+\cdots+b_{m}+(n-M)}\]
    for these values; and for $M+a_{m+1}+1\leq n\leq M+b_{m+1}$ we have $f(n)=n+a_{m}$, so for these values
    \[\frac{|\specn{\T}|}{n}=\frac{a_{0}+\cdots+a_{m}+a_{m+1}}{n}=\frac{a_{0}+\cdots+a_{m}+a_{m+1}}{b_{0}+\cdots+b_{m}+(n-M-a_{m+1})},\]
    and in both cases we obtain the aforementioned partial mediants.
\end{proof}

\section{\tp{Proof of \Cref{FW density general case}}{Proof of result \ref{FW density general case}}}

\FWdensitygeneralcase*

\begin{proof}
We divide the proof in two big cases: when $r$ is not computable, and when it is.
\begin{enumerate}
    \item Of course $0<r<1$, as both $0$ and $1$ are computable: because $r$ can be written in decimal notation as $0.0\cdots0d_{0}d_{1}d_{2}\cdots$, for digits $d_{i}\in\{0,1,\ldots, 9\}$ and $d_{0}\neq 0$, we can write $r$ as the limit of $\fr{a_{n}}{b_{n}}$, where $a_{n}=d_{0}\cdots d_{n}$, $b_{n}=10^{M+n}$, and $M$ is the number of zeros before $d_{0}$ (including the one before the decimal dot); this way, $0<a_{n}<b_{n}$ for all $n\in\mathbb{N}$. By the fact that $r$ is not computable, and since $\{b_{n}\}_{n\in\mathbb{N}}$ is certainly computable, $\{a_{n}\}_{n\in\mathbb{N}}$ cannot be so. We take the function $f$ associated with $\{a_{n}\}_{n\in\mathbb{N}}$ and $\{b_{n}\}_{n\in\mathbb{N}}$, as in \Cref{definition of f}, and the theory $\T$ axiomatized by $\{\psi_{\geq f(n+1)}\vee\bigvee_{i=1}^{n}\psi_{=f(i)} : n\in\No \}$: the theory is not strongly finitely witnessable because it is not smooth (see \cite[Theorem~7]{CADE}); furthermore, since $\{b_{n}\}_{n\in\mathbb{N}}$ is computable but $\{a_{n}\}_{n\in\mathbb{N}}$ is not, $f$ is also not computable, and since $\spec{\T}=\{f(n) : n\in\No \}$, we have that $\T$ does not have a computable minimal model function (\Cref{cmmf and computability}).

    Is is clear that $\mu(\T)$ is the limit of the mediants $\fr{(a_{0}+\cdots+a_{n})}{(b_{0}+\cdots+b_{n})}$, which from \Cref{mediant} is $\LIM \fr{a_{n}}{b_{n}}=r$; so we only have left to prove that $\T$ has a witness. For a quantifier-free formula $\phi$, let $N$ be the smallest $n\geq M$ such that $10^{N}\geq|\vars(\phi)|$, define $N(\phi)=\sum_{i=M}^{N}10^{i}$, and
    \[\wit(\phi)=\phi\wedge\bigwedge_{i=1}^{N(\phi)}x_{i}=x_{i}\]
    where the $x_{i}$ are fresh variables: we claim this function is a witness, obviously mapping quantifier-free formulas into quantifier-free formulas, and being computable, given that finding $|\vars(\phi)|$ and $N$ are both computable procedures. Furthermore, for $\overarrow{x}=\vars(\wit(\phi))\setminus\vars(\phi)$, it is clear that $\Exists{\overarrow{x}}\wit(\phi)$ and $\phi$ are $\T$-equivalent, given that in fact $\phi$ and $\wit(\phi)$ are themselves equivalent. Suppose then that $\A$ is a $\T$-interpretation that satisfies $\phi$, take a set $X$ disjoint from $\dom{\A}$ with cardinality $N(\phi)-|\vars(\phi)^{\A}|$, and we construct an interpretation $\B$ by making: $\dom{\B}=\vars(\phi)^{\A}\cup X$ (so $|\dom{\B}|=N(\phi)$, which is a cardinality in $\spec{\T}$); $x^{\B}=x^{\A}$ for all $x\in\vars(\phi)$; $x\in\{x_{i} : 1\leq i\leq N(\phi)\}\mapsto x^{\B}$ an injective function; and $x^{\B}$ can be set arbitrarily for all other variables. Then it is true that $\B$ satisfies $\wit(\phi)$ and $\dom{\B}=\vars(\wit(\phi))^{\B}$, finishing the proof.

\item Because $0\leq r\leq 1$ is computable, there are computable sequences $\{a_{n}\}_{n\in\mathbb{N}}$ and $\{b_{n}\}_{n\in\mathbb{N}}$ such that $0<a_{n}<b_{n}$ and $\LIM \fr{a_{n}}{b_{n}}=r$. Now, take a non-computable function $g:\No\rightarrow\{0,1\}$. We then define a function $\mf:\No\rightarrow\{0,1\}$ by induction.
    \begin{enumerate}
        \item 
        \begin{enumerate}
          \item $\mf(1)=g(1)$.
          \item $\mf(n)=0$ for: $2\leq n\leq 2(b_{0}-a_{0})$ (since $b_{0}>a_{0}$, $2(b_{0}-a_{0})\geq 2$) if $g(1)=0$; and $2\leq n\leq 2(b_{0}-a_{0})+1$ if $g(1)=1$.
          \item $\mf(n)=1$ for: $2(b_{0}-a_{0})+1\leq n\leq 2b_{0}$ if $g(1)=0$; and $2(b_{0}-a_{0})+2\leq n\leq 2b_{0}$ (since $b_{0}>a_{0}$, $2b_{0}\geq 2(b_{0}-a_{0})+2$) if $g(1)=1$.
        \end{enumerate}
        \item Now, assume $\mf(n)$ has been defined for $1\leq n\leq 2\sum_{i=0}^{m}b_{i}$.
        \begin{enumerate}
            \item $\mf(1+2\sum_{i=0}^{m}b_{i})=g(m+2)$;
            \item $\mf(n)=0$ for: $2+2\sum_{i=0}^{m}b_{i}\leq n\leq 2(b_{m+1}-a_{m+1})+2\sum_{i=0}^{m}b_{i}$ if $g(m+2)=0$; and $2+2\sum_{i=0}^{m}b_{i}\leq n\leq 1+2(b_{m+1}-a_{m+1})+2\sum_{i=0}^{m}b_{i}$ if $g(m+2)=1$.
            \item $\mf(n)=1$ for: $1+2(b_{m+1}-a_{m+1})+2\sum_{i=0}^{m}b_{i}\leq n\leq +2\sum_{i=0}^{m+1}b_{i}$ if $g(m+2)=0$; and $2+2(b_{m+1}-a_{m+1})+2\sum_{i=0}^{m}b_{i}\leq n\leq +2\sum_{i=0}^{m+1}b_{i}$ if $g(m+2)=1$.
        \end{enumerate}
    \end{enumerate}
    We see that $\mf$ is not computable because: the function $\rho:\No\rightarrow\No$ given by $\rho(1)=1$ and $\rho(m)=1+2\sum_{i=0}^{m-1}b_{i}$, for $m>1$, is computable, given that $\{b_{n}\}_{n\in\mathbb{N}}$ is computable; and, despite that, $\mf\circ\rho(m)=g(m)$, which is not computable.

    We can finally define the $\Sigma_{1}$-theory by the axiomatization 
    \[\{\psi_{\geq f(n+1)}\vee\bigvee_{i=1}^{n}\psi_{=f(i)} : n\in\No \},\]
    where $f:\No\rightarrow\No$ is defined as $f(n)=\min\{m\geq n : \mf(m)=1\}$, so that $\spec{\T}=\{f(n) : n\in\No\}$, and $|\specn{\T}|=|\{1\leq i\leq n : \mf(i)=1\}|$: this means $|\Spec_{2\sum_{i=0}^{m}b_{i}}(\T)|=2\sum_{i=0}^{m}a_{i}$, leading us again to partial mediants; from \Cref{mediant} we get $\mu(\T)$ equals $\LIM \fr{2a_{n}}{2b_{n}}$, which is precisely $r$. From the fact that $\mf$ is not computable it follows that $\spec{\T}$ is not computable, so from \Cref{cmmf and computability} $\T$ does not have a computable minimal model function. It only remains for us to show that $\T$ has a witness: so, for a quantifier-free formula $\phi$, let $N(\phi)$ be the minimum of $2\sum_{i=0}^{m}b_{i}$ greater than or equal to $|\vars(\phi)|$; this value is computable given the sequences $\{a_{n}\}_{n\in\mathbb{N}}$ and $\{b_{n}\}_{n\in\mathbb{N}}$ are themselves computable. For $x_{i}$ fresh variables, we define 
    \[\wit(\phi)=\phi\wedge\bigwedge_{i=1}^{N(\phi)}x_{i}=x_{i}:\]
    this obviously maps quantifier-free formulas into themselves, and is computable given that finding $N(\phi)$ from $\phi$ can be done algorithmically. It is also clear that $\phi$ and $\Exists{\overarrow{x}}\wit(\phi)$ are $\T$-equivalent, for $\overarrow{x}=\vars(\wit(\phi))\setminus\vars(\phi)$, since $\wit(\phi)$ is the conjunction of $\phi$ and a tautology. Finally, suppose that $\A$ is a $\T$-interpretation that satisfies $\phi$, and take a set $X$ disjoint from $\dom{\A}$ with $N(\phi)-|\vars(\phi)^{\A}|$ elements. We define an interpretation $\B$ by making: $\dom{\B}=\vars(\phi)^{\A}\cup X$ (so $|\dom{\B}|=N(\phi)$, and since $N(\phi)=2\sum_{i=0}^{m}b_{i}$ and $\mf(2\sum_{i=0}^{m}b_{i})=1$, we obtain $\B$ is a $\T$-interpretation); $x^{\B}=x^{\A}$ for all $x\in\vars(\phi)$; $x_{i}\in\{x_{i} : 1\leq i\leq N(\phi)\}\mapsto x_{i}^{\B}\in \dom{\B}$ an injective map; and arbitrarily for all other variables. We then have that $\B$ is a $\T$-interpretation that satisfies $\phi$, and $\dom{\B}=\vars(\wit(\phi))^{\B}$.
\end{enumerate}
\end{proof}

\section{\tp{Proof of \Cref{SI FMP non empty}}{Proof of result \ref{SI FMP non empty}}}

\SIFMPnonempty*

\begin{proof}
    Suppose $\T$ is not stably infinite, and thus there exists a quantifier-free formula $\phi$ such that no infinite $\T$-interpretation satisfies $\phi$. 
    This means there exists an $M\in\mathbb{N}$ such that no $\T$-interpretation $\A$ that satisfies $\phi$ has $|\dom{\A}|>M$: 
    were that not true, we could obtain a sequence of $\T$-interpretations $\A_{n}$ that satisfy $\phi$, with $|\dom{\A_{n}}|\geq n$; since $\A_{n}$ satisfies the axiomatization of $\T$, $\phi$ and $\{\psi_{\geq 1}, \ldots , \psi_{\geq n}\}$, we get by \Cref{compactness} that the axiomatization of $\T$, $\phi$ and $\{\psi_{\geq n}:n\in \mathbb{N}\}$ are simultaneously satisfiable, leading to a contradiction. So $\specn{\T,\phi}=\Spec_{M}(\T,\phi)$ for $n\geq M$, giving us $\mu(\T,\phi)\leq \LIM\fr{M}{n}=0$, and the result for stable infiniteness follows from the contrapositive.

    If $\T$ does not have the finite model property, there exists a quantifier-free formula $\phi$ such that no finite $\T$-interpretation satisfies $\phi$: that means $\spec{\T,\phi}=\emptyset$, and so $\mu(\T,\phi)=0$, the result for the finite model property easily following.
\end{proof}

\section{\tp{Proof of \Cref{smooth non empty}}{Proof of result \ref{smooth non empty}}}

\smoothnonempty*

\begin{proof}
    If $\T$ is smooth and has the finite model property, for every quantifier-free $\T$-satisfiable formula $\phi$ there exists $M\in\mathbb{N}$ such that for all $n\geq M$ there exists a $\T$-interpretation $\A$ that satisfies $\phi$ with $|\dom{\A}|=n$. That means $[1,n]\setminus\specn{\T,\phi}\subseteq [1,M]$, for $n\geq M$, and so $\mu(\T,\phi)\geq \LIM\fr{(n-M)}{n}=1$.
\end{proof}

\section{\tp{Proof of \Cref{SFW non empty}}{Proof of result \ref{SFW non empty}}}

The following technical lemma is a version of \cite[Theorem~7]{CADE}, modified to hold locally, that is, for every quantifier-free formula.

\begin{lemma}\label{CADE like result}
    If $\T$ is a strongly finitely witnessable theory, for all quantifier-free formulas $\phi$, we have that: 
    either $\phi$ is not satisfied by any infinite $\T$-in\-ter\-pre\-ta\-tions;
    or there exists $M\in\mathbb{N}$ such that, for all cardinals $M\leq\kappa\leq\aleph_{0}$, there is a $\T$-interpretation $\A$ that satisfies $\phi$ with $|\dom{\A}|=\kappa$.
\end{lemma}

\begin{proof}
    Suppose that $\phi$ is satisfied by an infinite $\T$-interpretation $\A$: We may assume that $|\dom{\A}|=\aleph_{0}$ by \Cref{LowenheimSkolem}.
    Since $\phi$ and $\Exists{\overarrow{x}}\wit(\phi)$ are $\T$-equivalent, for $\overarrow{x}=\vars(\wit(\phi))\setminus\vars(\phi)$, we may change the value assigned to some variables by $\A$ in order to obtain a $\T$-interpretation $\A^{\prime}$ that satisfies $\wit(\phi)$ which still has $|\dom{\A^{\prime}}|=\aleph_{0}$. 
    Let $U$ be $\vars(\wit(\phi))$, and $\delta_{U}$ the arrangement on $U$ such that $x$ is associated with $y$ iff $x^{\A^{\prime}}=y^{\A^{\prime}}$, meaning that $\A^{\prime}$ satisfies $\delta_{U}$.
    From that, $\A^{\prime}$ satisfies $\wit(\phi)\wedge\delta_{U}$, meaning there is a $\T$-interpretation $\B$ (from the fact that $\wit$ is a strong witness) that satisfies $\wit(\phi)\wedge\delta_{U}$, and thus $\phi$, with $\dom{\B}=\vars(\wit(\phi)\wedge\delta_{U})^{\B}=U^{\B}$.

    We state it is now possible to find a $\T$-interpretation $\C$ that satisfies $\phi$ with any cardinality $\kappa$ between $M=|\dom{\B}|$ and $\aleph_{0}$. 
    Indeed, if $\kappa=M$ we are done. 
    Otherwise, take a fresh set of variables $U^{\prime}$ with cardinality $\kappa-M$, and define an arrangement $\delta_{V}$ on $V=U\cup U^{\prime}$ such that $x$ is related to $y$ according to $\delta_{V}$ iff either $x=y$, or $x,y\in U$ and $x^{\B}=y^{\B}$;
    notice we can write 
    \[\delta_{V}=\delta_{U}\wedge\delta_{U^{\prime}}\wedge\bigwedge_{x\in U, y\in U^{\prime}}\neg(x=y),\]
    where $\delta_{U^{\prime}}$ corresponds to the identity on $U^{\prime}$. 
    Notice now that $\wit(\phi)\wedge\delta_{V}$ is $\T$-satisfiable: 
    indeed, $\A^{\prime}$ satisfies $\wit(\phi)\wedge\delta_{U}$, and by changing at most the values assigned by $\A^{\prime}$ to the (fresh) variables in $U^{\prime}$, we get a $\T$-interpretation $\A^{\prime\prime}$ that satisfies $\wit(\phi)\wedge\delta_{V}$.

    This means, again from the fact that $\wit$ is a strong witness, that there is a $\T$-interpretation $\C$ that satisfies $\wit(\phi)\wedge\delta_{V}$, with $\dom{\C}=\vars(\wit(\phi)\wedge\delta_{V})^{\C}=V^{\C}$.
    Since $\C$ satisfies $\wit(\phi)$, it satisfies $\Exists{\overarrow{x}}\wit(\phi)$ and thus $\phi$. 
    And since $\dom{\C}=V^{\C}$, where $U^{\C}$ must have $M$ elements (since $\C$ satisfies $\delta_{U}$), and $(U^{\prime})^{\C}$ must have another $\kappa-M$ elements (since $\C$ satisfies $\delta_{V}$), we get $|\dom{\C}|=\kappa$, as we wanted to show.
    
\end{proof}

\SFWnonempty*

\begin{proof}
     By \Cref{CADE like result}, either $\phi$ is not satisfied by an infinite $\T$-interpretation, or $\spec{\T,\phi}$ contains all elements greater than or equal to some $M$: from the second case it clearly follows that $\mu(\T,\phi)\geq\LIM\fr{(n-M+1)}{n}=1$. 
     In the first, we state $\spec{\T,\phi}$ must be finite: otherwise we should be able to get an infinite $\T$-interpretation that satisfies $\phi$ by using \Cref{compactness}.
     This way, $\spec{\T,\phi}$ must be bounded by some $M$, and so $\mu(\T,\phi)\leq\LIM\fr{M}{n}=0$.
\end{proof}

\section{\tp{Proof of \Cref{cmmf and computability gen}}{Proof of result \ref{cmmf and computability gen}}}

\cmmfcomputabilitygen*

\begin{proof}
     Suppose that $\T$ has a computable minimal model function $\minmod_{\T}$, and we shall consider two cases: 
     one where $\spec{\T,\phi}$ is finite, and one where it is infinite. 
     If $\spec{\T,\phi}$ is finite we have nothing left to do, as it is enough to simply hardcode these values into an algorithm that decides whether an element is in $\spec{\T,\phi}$. 
     If $\spec{T,\phi}$ is instead infinite, $\phi$ is $\T$-satisfiable (as $\spec{\T,\phi}$ is not empty) and, in addition, the formulas $\phi\wedge\NNEQ{x}$ are $\T$-satisfiable for all $n\in\mathbb{N}$: 
     indeed, since $\spec{\T,\phi}$ is infinite, we can always find $m\in\spec{\T,\phi}$ such that $m\geq n$, and thus there exists a $\T$-interpretation $\A$ that satisfies $\phi$ with $|\dom{\A}|=m$; 
     by changing at most the value assigned to the variables $x_{1},\ldots,x_{n}$ (assumed fresh), we get a $\T$-interpretation that satisfies $\phi\wedge\NNEQ{x}$. 
     Given then $n$, we define $f(0)=\minmod_{\T}(\phi)$, and assuming $f(m)$ defined $f(m+1)=\minmod_{\T}(\phi\wedge\NNNEQ{x}{f(m)+1})$, and then we state that $n\in\spec{\T,\phi}$ iff $n\in\{f(0),\ldots,f(n)\}$. 
    
    That this results in a decision procedure follows from the fact that $\minmod_{\T}$ is assumed to be computable, and so is producing the formulas $\NNNEQ{x}{f(m)+1}$; 
    we have left to prove that it is both sound and complete. 
    If $n\in\{f(0),\ldots,f(n)\}$, there exists an $0\leq m< n$ such that $n=\minmod_{\T}(\phi\wedge\NNNEQ{x}{f(m)+1})$, and so there is a $\T$-interpretation $\A$ that satisfies $\phi\wedge\NNNEQ{x}{f(m)+1}$ with $|\dom{\A}|=n$; of course $\A$ is a $\T$-interpretation that satisfies $\phi$ with $|\dom{\A}|=n$, so $n\in\spec{\T,\phi}$ and indeed the algorithm is sound. 
    Reciprocally, if $n\in\spec{\T,\phi}$, suppose that $f(m)<n<f(m+1)$ for some $0\leq m<n$ (there is either such an element or $n=f(m)$ and we have nothing to prove, since $f(m+1)\geq f(m)+1$): 
    because $n$ is in $\spec{\T,\phi}$, there is a $\T$-interpretation $\A$ that satisfies $\phi$ with $|\dom{\A}|=n$; 
    but since $n>f(m)$, $n\geq f(m)+1$, and so $\A$ satisfies $\phi\wedge\NNNEQ{x}{f(m)+1}$. 
    This, in combination with the fact that $f(m+1)>n$, contradicts that $f(m+1)=\minmod_{\T}(\phi\wedge\NNNEQ{x}{f(m)+1})$, proving that the algorithm is also complete and finishing the proof.
\end{proof}

\section{\tp{Proof of \Cref{CMMF non empty}}{Proof of result \ref{CMMF non empty}}}

\CMMFnonempty*

\begin{proof}
If $\T$ has a minimal model function, \Cref{cmmf and computability gen} guarantees all $\spec{\T,\phi}$ are computable, and by \Cref{comp. set comp. number} it follows that all $\mu(\T,\phi)$ are computable. 

Finally, for every computable $r$, we get a theory $\T$ with a computable minimal model function and $\mu(\T)=r$ from \Cref{cmmf density}, and it follows that for any tautology $\phi$ it is true that $\mu(\T,\phi)=r$.

\end{proof}

\section{\tp{Proof of \Cref{FW non empty}}{Proof of result \ref{FW non empty}}}

\FWnonempty*

\begin{proof}
    Follows trivially from \Cref{FW density general case}: taking, for any $0\leq r\leq 1$, a finitely witnessable theory $\T$ (without a computable minimal model function) with $\mu(\T)=r$, and a tautology $\phi$, $\mu(\T,\phi)=r$.
\end{proof}

\end{document}